\documentclass[11pt]{article}
\usepackage{algorithm}
\usepackage{algorithmic}

\usepackage{subfigure}
\usepackage{epsfig,amsthm,amsmath,color, amsfonts}
\usepackage{epsfig,color}

\usepackage{framed}
\newtheorem{theorem}{Theorem}[section]

\newtheorem{proposition}[theorem]{Proposition}
\newtheorem{lemma}[theorem]{Lemma}
\newtheorem{claim}[theorem]{Claim}

\theoremstyle{definition}
\theoremstyle{definition}\newtheorem{definition}[theorem]{Definition}
\theoremstyle{observation}

\newcommand{\comment}[1]{}
\newcommand{\QED}{\mbox{}\hfill \rule{3pt}{8pt}\vspace{10pt}\par}

\def\m{{\rm min}}

\def\polylog{\operatorname{polylog}}
\newcommand{\ignore}[1]{}
\newcommand{\eat}[1]{}

\usepackage{times}
\usepackage[small,compact]{titlesec}
\usepackage[small,it]{caption}

\newcommand{\squishlist}{
 \begin{list}{$\bullet$}
  { \setlength{\itemsep}{0pt}
     \setlength{\parsep}{3pt}
     \setlength{\topsep}{3pt}
     \setlength{\partopsep}{0pt}
     \setlength{\leftmargin}{1.5em}
     \setlength{\labelwidth}{1em}
     \setlength{\labelsep}{0.5em} } }
\newcommand{\squishend}{
  \end{list}  }

\def\e{{\rm E}}

\def\bone{{\bf 1}}

\def\danupon#1{\marginpar{$\leftarrow$\fbox{D}}\footnote{$\Rightarrow$~{\sf #1 --Danupon}}}
\def\prasad#1{}
\def\gopal#1{}
\def\atish#1{}

\begin{document}

\title{Efficient Distributed Random Walks with Applications
}

\begin{titlepage}
\author{Atish {Das Sarma} \thanks{College of Computing, Georgia Institute of Technology, Atlanta, GA 30332, USA.
\hbox{E-mail}:~{\tt atish@cc.gatech.edu, danupon@cc.gatech.edu}} \and Danupon Nanongkai \addtocounter{footnote}{-1}
\footnotemark \and  Gopal Pandurangan \thanks{Division of Mathematical
Sciences, Nanyang Technological University, Singapore 637371 and Department of Computer Science, Brown University, Providence, RI 02912.  \hbox{E-mail}:~{\tt gopalpandurangan@gmail.com}. Supported in part by NSF grant CCF-0830476.}   \and Prasad Tetali \thanks{School of Mathematics and School of Computer Science,
Georgia Institute of Technology
Atlanta, GA 30332, USA. \hbox{E-mail}:~{\tt tetali@math.gatech.edu}. Supported in part by NSF DMS 0701023 and NSF CCR 0910584.}}

\date{}

\maketitle \thispagestyle{empty}

\vspace*{.4in}

\maketitle
\begin{abstract}
We focus on  the problem of performing random walks efficiently in a distributed network. Given bandwidth constraints, the goal is to minimize the number of rounds required to obtain a random walk sample. We first present a fast sublinear time distributed algorithm for performing random walks whose time complexity is sublinear in the length of the walk. Our algorithm performs a random walk of length $\ell$  in $\tilde{O}(\sqrt{\ell D})$  rounds (with high probability) on an undirected  network, where $D$ is the diameter of the network. This improves over the previous best algorithm that ran in $\tilde{O}(\ell^{2/3}D^{1/3})$ rounds (Das Sarma et al., PODC 2009). We further extend our algorithms to efficiently perform $k$ independent random walks in   $\tilde{O}(\sqrt{k\ell D} + k)$ rounds. We then show that there is a fundamental difficulty in improving the dependence on $\ell$ any further by proving a lower bound of $\Omega(\sqrt{\frac{\ell}{\log \ell}} + D)$ under a general model of distributed random walk algorithms. Our random walk algorithms are useful in speeding up distributed algorithms for a variety of applications that use random walks as a subroutine. We present two main applications. First, we give a fast distributed algorithm for computing a random spanning tree (RST) in an arbitrary (undirected) network which runs in $\tilde{O}(\sqrt{m}D)$ rounds (with high probability; here $m$ is the number of edges). Our second application is a fast decentralized algorithm for estimating mixing time and related parameters of the underlying network. Our algorithm is fully decentralized and can serve as a building block in the design of topologically-aware networks.

\end{abstract}




\end{titlepage}

\vspace{-0.15in}
\section{Introduction}
Random walks play a central role in computer science, spanning a
wide range of areas in both theory and practice. The focus  of this
paper is  random walks in networks, in particular, decentralized
algorithms for performing random walks in arbitrary networks. Random
walks are used as an integral subroutine in a wide variety of
network applications ranging from token management and load
balancing to search, routing, information propagation and gathering,
network topology construction and building random spanning trees
(e.g., see \cite{DNP09-podc} and the references therein). Random
walks  are also very useful in providing uniform and efficient
solutions to distributed control of dynamic networks \cite{BBSB04,
ZS06}.  Random walks  are local and lightweight and require little
index or state maintenance which make them especially attractive to
self-organizing dynamic networks such as Internet overlay and ad hoc
wireless networks.

A key purpose of random walks in  many of these network applications
is to perform  node sampling.  While the sampling requirements in different
applications vary, whenever a true sample is required from a random
walk of certain steps, typically all applications perform the walk naively
--- by simply passing a token from one node to its neighbor: thus to
perform a random walk of length $\ell$ takes time linear in $\ell$.

In this paper, we  present a  sublinear  time (sublinear in $\ell$) distributed  random walk
sampling algorithm that is significantly  faster than the previous
best result. Our algorithm runs in time $\tilde{O}(\sqrt{\ell D})$ rounds.
  We then present  an almost matching lower bound that applies
to a general class  of distributed algorithms (our algorithm also falls in this class).
Finally, we present two key applications of our algorithm.
The first is a fast distributed algorithm for computing a random spanning tree, a fundamental spanning tree problem that has been studied widely in the classical setting (see e.g., \cite{kelner-madry} and references therein). To the best of our knowledge,  our algorithm gives the fastest known running time in an arbitrary network.
 The second is to devising efficient decentralized algorithms for computing
 key global metrics of the underlying network ---
 mixing time, spectral gap, and conductance. Such algorithms can be useful building
 blocks in the design of {\em topologically (self-)aware} networks, i.e., networks that can  monitor and regulate themselves in a decentralized fashion. For example,  efficiently computing the mixing time or the spectral gap, allows  the network to monitor connectivity and expansion properties of the network.

\subsection{Distributed Computing Model}
Consider an undirected, unweighted, connected $n$­-node graph $G =
(V, E)$.  Suppose that every node (vertex) hosts a processor with
unbounded computational power, but with limited initial knowledge.
Specifically, assume that each node is associated with a distinct identity
number from the set $\{1, 2, . . . , n\}$. At the beginning of the
computation, each node $v$ accepts as input its own identity number
and the identity numbers of its neighbors in $G$. The node may also
accept some additional inputs as specified by the problem at hand.
The nodes are allowed to communicate through the edges of the graph
$G$. The communication is synchronous, and occurs in discrete
pulses, called {\em rounds}. In particular, all the nodes wake up
simultaneously at the beginning of round 1, and from this point on
the nodes always know the number of the current round. In each round
each node $v$ is allowed to send an arbitrary message of size
$O(\log n)$ through each edge $e = (v, u)$ that is adjacent to $v$,
and the message will arrive to $u$ at the end of the current round.
This is a standard model of distributed computation known as the
{\em CONGEST model} \cite{peleg} and has been attracting a lot of
research attention during last two decades
(e.g., see \cite{peleg} and the references therein).

There are several measures of efficiency of distributed algorithms,
but we will concentrate on one of them, specifically, {\em the
running time}, that is, the number of rounds of distributed
communication. (Note that the computation that is performed by the
nodes locally is ``free'', i.e., it does not affect the number of rounds.)
 Many
fundamental network problems such as minimum spanning tree, shortest
paths, etc. have been addressed in this model (e.g., see
\cite{lynch, peleg, PK09}). In particular, there has been much
research into designing very fast distributed   approximation
algorithms (that are even faster at the cost of producing
sub-optimal solutions) for many of these  problems (see e.g.,
\cite{elkin-survey,dubhashi, khan-disc,khan-podc}).  Such algorithms
can be useful for large-scale resource-constrained and
dynamic networks where running time is crucial.

\subsection{Problem Statement, Motivation, and Related Work}
The basic problem we address is the following.
We are given an arbitrary undirected, unweighted, and connected $n$--node
network $G = (V,E)$ and a (source) node $s \in V$.
The goal is to devise a distributed algorithm such that, in the end,
$s$ outputs the ID of a node $v$ which is randomly picked according
to the probability that it is the destination of a random walk of
length $\ell$ starting at $s$. Throughout this paper, we assume
the standard (simple) random walk: in each step, an edge is taken from
the current node $x$ with probability proportional to $1/d(x)$ where
$d(x)$ is the degree of $x$. Our goal is to output a true  random
sample from the $\ell$-walk distribution starting from $s$.

For clarity, observe that the following naive algorithm solves the
above problem in $O(\ell)$ rounds:
The walk of length $\ell$ is performed by sending a token for $\ell$
steps, picking a random neighbor with each step. Then, the
destination node $v$ of this walk sends its ID back (along the same path) to the source for output. Our goal is to perform such
sampling with significantly less number of rounds, i.e., in time that
is sublinear in $\ell$.  On the other hand, we note that it can take too much time (as much as $\Theta(|E|+D)$ time) in the CONGEST model
to collect all  the topological information at
the source node (and then computing the walk locally).

This problem was proposed in~\cite{DNP09-podc} under the name \textit{Computing One Random Walk where Source Outputs
Destination (1-RW-SoD)}
(for short, this problem will be simply called {\em Single Random
Walk} in this paper),
wherein the first sublinear time distributed algorithm was provided,
requiring $\tilde{O}(\ell^{2/3}D^{1/3})$ rounds ($\tilde{O}$
hides $\polylog(n)$ factors); this improves over the naive $O(\ell)$
algorithm when the walk is long compared to the diameter (i.e.,
$\ell = \Omega(D \polylog n)$ where $D$ is the diameter of the network).
This was the first
result to break past the inherent sequential nature of random walks and beat
the naive $\ell$ round
approach, despite the fact that random walks have been used in
distributed networks for long and in a wide variety of applications.

There are two key motivations for obtaining sublinear time bounds.
The first is that in many algorithmic applications, walks of length significantly greater than the network diameter are needed.
For example, this is necessary in both the  applications   presented later in the paper, namely distributed computation of a random spanning tree (RST) and  computation of mixing time. In the RST algorithm, we need to perform a random walk
of expected length $O(mD)$ (where $m$ is the number of edges in the network).
In decentralized computation of mixing time, we need to perform walks
of length at least equal to the mixing time which can be significantly larger than the diameter (e.g., in a random geometric graph model \cite{MP}, a popular model
for ad hoc networks, the mixing time can be larger than the diameter by a
factor of $\Omega(\sqrt{n})$.)
More generally,
many real-world communication networks  (e.g., ad hoc  networks and peer-to-peer networks) have relatively small diameter, and random walks of
length at least the diameter are usually performed for many sampling applications, i.e., $\ell >> D$. It should be noted that  if the network is rapidly mixing/expanding which is sometimes the case in practice, then sampling from walks of length $\ell >> D$ is close to sampling from the steady state (degree) distribution; this can be done in $O(D)$ rounds (note however, that this gives only an approximately close sample, not the exact sample for that length). However, such an approach fails when $\ell$ is smaller than the mixing time.

The second motivation is understanding
the time complexity of distributed random walks. Random walk is essentially a global problem  which requires the algorithm to ``traverse" the entire network.
Classical ``global" problems include the minimum spanning tree, shortest path etc. Network diameter is an inherent lower bound for such problems. Problems of this type raise the basic question whether $n$ (or $\ell$ as the case here) time is essential or is the network diameter $D$, the inherent parameter. As pointed out in the seminal work of \cite{peleg-mst}, in the latter case, it would be
desirable to design algorithms that have a better complexity for graphs with low
diameter.

%
%

The high-level idea used in the $\tilde{O}(\ell^{2/3}D^{1/3})$-round
algorithm in \cite{DNP09-podc} is to ``prepare'' a few short walks
in the beginning (executed in parallel) and then carefully stitch
these walks together later as necessary.
The same general approach was introduced in~\cite{AtishGP08} to find random walks in data streams with the
main motivation of finding PageRank.
However, the two models have very different constraints and
motivations and hence the subsequent techniques used in \cite{DNP09-podc} and \cite{AtishGP08} are very different.
%

%
%

%
Recently, Sami and Twigg~\cite{ST08} consider lower bounds on the
communication complexity of computing stationary distribution of
random walks in a network. Although, their problem is  related to
our problem, the lower bounds obtained do not  imply anything in our
setting. Other recent works involving multiple random walks in different settings include Alon
et. al.~\cite{AAKKLT}, and Cooper et al. \cite{frieze}.
%

\subsection{Our Results}

\squishlist
\item {\bf A Fast Distributed Random Walk Algorithm:} We present a sublinear, almost time-optimal, distributed algorithm for the single random walk
problem in arbitrary networks that runs in  time $\tilde{O}(\sqrt{\ell D})$, where $\ell$ is the length of the walk (cf. Section \ref{sec:upperbound}). This is a
significant improvement over the naive $\ell$-round algorithm for $\ell = \Omega(D)$ as well
as  over the previous best running time of
$\tilde{O}(\ell^{2/3}D^{1/3})$  \cite{DNP09-podc}. The dependence on $\ell$
is reduced from $\ell^{2/3}$ to $\ell^{1/2}$.


Our algorithm in this paper uses an approach similar to that of
\cite{DNP09-podc} but  exploits certain key properties of random
walks  to design an even faster sublinear time algorithm. Our algorithm is randomized (Las Vegas type, i.e., it always  outputs the correct result, but the running  time claimed is with high probability) and is conceptually simpler compared to the  $\tilde{O}(\ell^{2/3}D^{1/3})$-round algorithm (whose running time is deterministic). While the previous (slower) algorithm \cite{DNP09-podc} applies to the more general Metropolis-Hastings walk, in this work we focus primarily on the simple random walk for the sake of obtaining the best possible bounds in this commonly used setting.
%

   One of the key ingredients in the improved algorithm is proving a bound on the number of times any node is visited in an $\ell$-length walk, for any  length $\ell = O(m^2)$. We show that w.h.p. any node $x$ is visited at most $\tilde{O}(d(x)\sqrt{\ell})$ times, in an $\ell$-length walk from any starting node ($d(x)$ is the degree of $x$).  We then show that if only certain $\ell/\lambda$ special  points of the walk (called as {\em connector points}) are observed,
then any node is observed only $\tilde{O}(d(x)\sqrt{\ell}/\lambda)$ times. The algorithm starts with all nodes performing short walks (of length uniformly random in the range $\lambda$ to $2\lambda$ for appropriately chosen $\lambda$) efficiently simultaneously; here the randomly chosen lengths play a crucial role in arguing about a suitable spread of the connector points.   Subsequently, the algorithm begins at the source and carefully stitches these walks together till $\ell$ steps are completed.
%

We also extend to give algorithms for computing $k$ random walks (from any $k$ sources
 ---not necessarily distinct) in $\tilde O\left(\min(\sqrt{k\ell D}+k, k+\ell)\right)$ rounds. Computing $k$ random
walks is useful in many applications such as the one we present below on
decentralized computation of mixing time and related parameters. While the main requirement of our algorithms is to just obtain the random walk samples (i.e. the end point of the $\ell$ step walk), our algorithms can regenerate the entire walks such that each node knows its position(s) among the $\ell$ steps.
Our algorithm
can  be extended to do this in the same number of rounds.

\item {\bf A Lower Bound:} We establish an almost matching lower bound on the running time of distributed random walk that applies to
a general class of distributed random walk algorithms. We show that any algorithm belonging to the class
needs at least
$\Omega(\sqrt{\frac{\ell}{\log \ell}} + D)$  rounds to perform a random walk of
length $\ell$; notice that this lower bound is nontrivial even in graphs of  small ($D = O(\log n)$)  diameter (cf. Section \ref{sec:lowerbound}).
 Broadly speaking,  we consider   a class of
 token forwarding-type algorithms where nodes can only store and (selectively) forward  tokens (here tokens are $O(\log n)$-sized messages consisting of two node ids  identifying the
beginning and end of a segment --- we make this more precise in Section \ref{sec:lowerbound}). Selective forwarding (more general than just store and forwarding) means that
nodes can omit to forward certain segments (to reduce number of messages), but they cannot alter   tokens in any way (e.g., resort
to data compression techniques). This class includes many natural algorithms, including the algorithm in this paper.

Our technique involves showing the same non-trivial lower bound  for a problem that we call {\em path verification}. This simpler problem appears quite basic and can have other applications. Informally, given a graph $G$ and a sequence of $\ell$ vertices in
the graph, the problem is for some (source) node in the graph to
verify that the sequence forms a path.
One main idea in this proof is to show that independent nodes may be able to verify short {\em local} paths; however, to be able to {\em merge} these together and verify an $\ell$-length path would require exchanging several messages. The trade-off is between the lengths of the local paths that are verified and the number of such local paths that need to be combined.  Locally verified paths can be exchanged in one round, and messages can be exchanged at all nodes.  Despite this, we show that the bandwidth restriction necessitates a large number of rounds even if the diameter is small.
We then show a  reduction to the random walk
problem, where we require that each node in the walk should know its (correct) position(s) in the walk.

Similar non-trivial matching lower bounds on running time  are known only for
a few important problems in distributed computing, notably the
minimum spanning tree problem (e.g., see \cite{peleg-bound, elkin}).
Peleg and Rabinovich \cite{peleg-bound} showed that
$\tilde{\Omega}(\sqrt{n})$ time is required for constructing an MST
even on graphs of small diameter (for any $D=\Omega(\log n)$) and \cite{kutten-domset} showed an essentially matching upper bound.

\item {\bf Applications:} Our faster distributed random walk algorithm can be used in speeding up
distributed applications where  random walks arise as a subroutine.
Such applications include distributed construction of expander graphs,
checking whether a graph is an expander, construction of random spanning trees, and random-walk based search (we refer to \cite{DNP09-podc} for details).
Here, we present two key applications:

(1) {\em A Fast Distributed Algorithm for Random Spanning Trees (RST):}
We give a $\tilde{O}(\sqrt{m}D)$ time distributed algorithm (cf. Section \ref{sec:rst}) for uniformly sampling a random spanning tree in an arbitrary undirected
(unweighted) graph (i.e., each spanning tree in the underlying network has the same probability of being selected). ($m$ denotes the number of edges in the graph.)
Spanning trees are fundamental network primitives
and distributed algorithms for various types of spanning trees such as minimum spanning tree (MST), breadth-first spanning tree (BFS), shortest path tree,
shallow-light trees etc., have been studied extensively in the literature \cite{peleg}. However, not much is known about the distributed complexity
of the random spanning tree problem.
The centralized case
has been studied for many decades, see e.g.,
the recent work of \cite{kelner-madry} and the references therein; also see the recent work of Goyal et al.
\cite{goyal} which gives nice applications of RST to fault-tolerant routing
and constructing expanders.  In the distributed context, the work
of Bar-Ilan and Zernik \cite{bar-ilan} give a distributed RST algorithm for
two  special cases, namely that of a complete graph (running in constant time) and a synchronous ring (running in  $O(n)$ time).
The work of \cite{Baala}
give a self-stablizing distributed algorithm for constructing a RST in a wireless ad hoc network and mentions that RST is more resilient to transient
failures that occur in mobile ad hoc networks.

Our algorithm works by giving an efficient distributed implementation
of the well-known Aldous-Broder random walk algorithm \cite{aldous, broder} for constructing a RST.

(2) {\em Decentralized Computation of Mixing Time.} We present a fast decentralized algorithm for estimating mixing time, conductance and spectral gap of the network (cf. \ref{sec:mixingtime}). In
particular, we show that given a starting point $x$, the mixing time with respect to $x$, called $\tau^x_{mix}$, can be
estimated in $\tilde{O}(n^{1/2} + n^{1/4}\sqrt{D\tau^x_{mix}})$ rounds. This gives an alternative algorithm to the only previously known
approach by Kempe and McSherry \cite{kempe} that can be used to estimate
$\tau^x_{mix}$ in $\tilde O(\tau^x_{mix})$ rounds.\footnote{Note
that \cite{kempe} in fact do more and give a decentralized algorithm for
computing the top $k$ eigenvectors of a weighted adjacency matrix
that runs in $O(\tau_{mix}\log^2 n)$ rounds if two adjacent nodes are allowed to exchange $O(k^3)$ messages per round, where $\tau_{mix}$ is
the mixing time and $n$ is the size of the network.}  To compare,  we note that
when $\tau^x_{mix} = \omega(n^{1/2})$ the present algorithm is faster (assuming
$D$ is not too large).


  The
work of \cite{mihail-topaware}  discusses spectral algorithms for
enhancing the topology awareness, e.g., by identifying and assigning
weights to critical links. However, the algorithms are centralized,
and it is mentioned that obtaining efficient decentralized
algorithms is a major open problem. Our algorithms are fully decentralized
and  based on performing random walks,
and so more amenable to dynamic and self-organizing networks.
\squishend


\section{A Sublinear Time Distributed Random Walk Algorithm}\label{sec:one_walk_DoS}
\label{sec:upperbound}
\subsection{Description of the Algorithm}
We first describe the
$\tilde{O}(\ell^{2/3}D^{1/3})$-round algorithm in \cite{DNP09-podc}
and then highlight the changes in our current algorithm. The current
algorithm is randomized and uses several new ideas that are crucial in obtaining the new bound.

The high-level idea is to perform ``many" short random walks
in parallel and later stitch them
together as needed (see Figure~\ref{fig:connector} in Appendix).
In the first phase of the algorithm {\sc
Single-Random-Walk} (we refer to Appendix for pseudocodes of all
algorithms and subroutines), each node performs $\eta$ independent
random walks of length $\lambda$. (Only the destination of each of
these walks is aware of its source, but the sources do not know
destinations right away.) It is shown that this takes
$\tilde{O}(\eta\lambda)$ rounds with high probability. Subsequently,
the source node that requires a walk of length $\ell$ extends a walk
of length $\lambda$ by ``stitching'' walks.  If the end point of the
first $\lambda$ length walk is $u$, one of $u$'s $\lambda$ length
walks is used to extend. When at $u$, one of its $\lambda$-length walk destinations are
sampled uniformly (to preserve randomness) using {\sc Sample-Destination} in $O(D)$ rounds.  (We call
such $u$ and other nodes at the stitching points as {\em connectors}
--- cf. Algorithm 1.) Each stitch takes $O(D)$ rounds (via the shortest path). This process is extended as long as unused
$\lambda$-length walks are available from visited nodes. If the walk
reaches a node $v$ where all $\eta$ walks have been used up (which is a key difficulty), then
{\sc Get-More-Walks} is invoked. {\sc Get-More-Walks} performs
$\eta$ more walks of length $\lambda$ from $v$, and this can be done in $\tilde{O}(\lambda)$ rounds. The number of
times {\sc Get-More-Walks} is invoked can be bounded by
$\frac{\ell}{\eta\lambda}$ in the worst case by an amortization argument.  The
overall bound on the algorithm is  $O(\eta\lambda + \ell D/\lambda +
\frac{\ell}{\eta})$. The bound of $\tilde{O}(\ell^{2/3}D^{1/3})$ follows from appropriate choice of parameters $\eta$ and $\lambda$.


The current algorithm uses two crucial ideas to improve the running time.
The first  idea is
 to bound the number of times any node is visited in a random walk of length
 $\ell$ (in other words, the number of times
{\sc Get-More-Walks} is invoked). Instead of the worst case analysis
in \cite{DNP09-podc}, the new bound is obtained by bounding the number
of times any node is visited (with high probability) in a random
walk of length $\ell$ on an undirected unweighted graph. The number
of visits to a node beyond the mixing time can be bounded using its
stationary probability distribution. However, we need a bound on the
visits to a node  for any $\ell$-length walk starting from the first step.
 We show a somewhat surprising bound that applies to an $\ell$-length (for $\ell = O(m^2)$)  random walk on any arbitrary (undirected) graph: {\em no node $x$ is visited more than
 $\tilde{O}(d(x)\sqrt{\ell})$ times}, in an $\ell$-length walk from any starting node ($d(x)$ is the degree of $x$)   (cf. Lemma~\ref{lemma:visits bound}).
 Note that this bound does not depend on any other parameter of the graph, just on the (local) degree of the node and the length of the walk. This bound is
 tight in general (e.g., consider a line and a walk of length $n$).

The above bound is not enough to get the desired running time, as it does not say anything about the
 distribution of connectors when we chop the length $\ell$ walk into $\ell/\lambda$ pieces.
 We have to bound the number
of visits to a node as a connector in order to bound the number of
times {\sc Get-More-Walks} is invoked. To overcome this we use a second idea:
  Instead of nodes performing walks of length $\lambda$, each such walk $i$ is of length $\lambda+r_i$ where $r_i$ is a random number in the range $[0,\lambda-1]$. Notice that the random numbers are independent for each walk.
   We show  the following ``uniformity lemma":  if  the short walks are now of a random length in the range of $[\lambda, 2\lambda-1]$, then if a node $u$ is visited at most $N_u$ times  in an $\ell$ step walk, then the node is visited at most $\tilde{O}(N_u/\lambda)$ times as an endpoint of a short walk  (cf. Lemma \ref{lem:uniformityused}).  This modification to {\sc Single-Random-Walk} allows us to
bound the number of visits to each node (cf.
Lemma~\ref{lem:uniformityused}).

%
%

The change of the short walk length above leads to two modifications
in Phase~1 of {\sc Single-Random-Walk} and {\sc Get-More-Walks}. In
Phase~1, generating $\eta$ walks of different lengths from each node
is straightforward: Each node simply sends $\eta$ tokens containing
the source ID and the desired length. The nodes keep forwarding
these tokens with decreased desired walk length until the desired
length becomes zero.
The modification of {\sc Get-More-Walks} is tricker. To avoid
congestion, we use the idea of {\em reservoir
sampling}~\cite{Vitter85}. In particular, we add the following
process at the end of the {\sc Get-More-Walks} algorithm in~\cite{DNP09-podc}:

\begin{quote}
\begin{algorithmic}
\FOR{$i=0$ to $\lambda-1$}

\STATE \label{line:reservoir} For each message, independently with
probability $\frac{1}{\lambda-i}$, stop sending the message further
and save the ID of the source node (in this event, the node with the
message is the destination). For messages $M$ that are not stopped,
each node picks a neighbor correspondingly and sends the messages
forward as before.

\ENDFOR
\end{algorithmic}
\end{quote}

%
%

The reason it needs to be done this way is that if we first sampled
the walk length $r$, independently for each walk, in the range
$[0,\lambda-1]$ and then extended each walk accordingly, the
algorithm would need to pass $r$ independently for each walk. This
will cause congestion along the edges; no congestion occurs in the
mentioned algorithm as only the {\em count} of the number of walks
along an edge are passed to the node across the edge. Therefore, we
need to decide when to stop on the fly using  reservoir sampling.

%

We also have to make another modification in Phase~1 due to the new
bound on the number of visits. Recall that, in this phase, each node
prepares $\eta$ walks of length $\lambda$. However, since the new
bound of visits of each node $x$ is proportional to its degree
$d(x)$ (see Lemma~\ref{lemma:visits bound}), we make each node
prepare $\eta d(x)$ walks instead. We show that Phase~1 uses $\tilde
O(\eta\lambda)$ rounds, instead of $\tilde O(\frac{\lambda
\eta}{\delta})$ rounds where $\delta$ is the minimum degree in the
graph (cf. Lemma~\ref{lem:Sample-Destination}).

To summarize, the main algorithm for performing a single random walk
is {\sc Single-Random-Walk}. This algorithm, in turn, uses {\sc
Get-More-Walks} and {\sc Sample-Destination}.
The key modification is that, instead of creating short walks of
length $\lambda$ each, we create short walks where each walk has
length in range $[\lambda, 2\lambda-1]$. To do this, we modify the
Phase~1 of {\sc Single-Random-Walk} and {\sc Get-More-Walks}.

We now state four lemmas which are  similar to the Lemma~2.2-2.6
in \cite{DNP09-podc}. However, since the algorithm
here is a modification of that in \cite{DNP09-podc}, we
include the full proofs in Appendix~\ref{sec:four-proofs}.


\begin{lemma} \label{lem:phase1}
Phase~1 finishes in $O(\lambda \eta \log{n})$ rounds with high
probability.
\end{lemma}



\begin{lemma}\label{lem:get-more-walks}
For any $v$, {\sc Get-More-Walks($v$, $\eta$, $\lambda$)} always
finishes within $O(\lambda)$ rounds.
\end{lemma}

\begin{lemma}\label{lem:Sample-Destination}
{\sc Sample-Destination} always finishes within $O(D)$ rounds.
\end{lemma}


\begin{lemma}\label{lem:correctness-sample-destination-new}
Algorithm {\sc Sample-Destination}($v$) (cf.
Algorithm~\ref{alg:Sample-Destination}) returns a destination from
a random walk whose length is uniform in the range $[\lambda,2\lambda-1]$.
\end{lemma}




\subsection{Analysis}
\label{sec:analysis}

The following theorem states the main result of this Section. It
states that the algorithm {\sc Single-Random-Walk} correctly samples
a node after a random walk of $\ell$ steps and the algorithm takes,
with high probability, $\tilde O\left(\sqrt{\ell D}\right)$ rounds
where $D$ is the diameter of the graph. Throughout this section, we assume that $\ell$ is $O(m^2)$, where $m$ is the number of edges in the network. If $\ell$ is $\Omega(m^2)$, the required bound is easily achieved by aggregating the graph topology (via upcast) onto one node in $O(m+D)$ rounds (e.g., see \cite{peleg}). The difficulty lies in proving for  $\ell = O(m^2) $.

\begin{theorem}\label{thm:1-walk}
For any $\ell$, Algorithm {\sc Single-Random-Walk} (cf.
Algorithm~\ref{alg:single-random-walk}) solves $1$-RW-DoS (the
Single Random Walk Problem) and, with probability at least
$1-\frac{2}{n}$,
finishes in $\tilde O\left(\sqrt{\ell D}\right)$ rounds.
\end{theorem}

We prove the above theorem using the following lemmas.
As mentioned earlier, to bound the number of times {\sc Get-More-Walks} is invoked, we  need a technical result on random walks that bounds
the number of times a node will be visited in a $\ell$-length random walk.
Consider a simple random walk on a connected undirected graph on $n$
vertices. Let $d(x)$ denote the degree of $x$, and let $m$ denote
the number of edges. Let $N_t^x(y)$ denote the number of visits to
vertex $y$ by time $t$, given the walk started at vertex
$x$.
Now, consider $k$ walks, each of length $\ell$, starting from (not
necessary distinct) nodes $x_1, x_2, \ldots, x_k$.
We show a key technical lemma (proof in Appendix ~\ref{sec:proof of visits bound}) that applies to a random walk on any graph:  With high probability, no vertex $y$ is visited more
than $24 d(x) \sqrt{k\ell+1}\log n + k$ times.

\begin{lemma}\label{lemma:visits bound}
For any nodes $x_1, x_2, \ldots, x_k$, and $\ell=O(m^2)$, \[\Pr\bigl(\exists y\ s.t.\
\sum_{i=1}^k N_\ell^{x_i}(y) \geq 24 d(x) \sqrt{k\ell+1}\log
n+k\bigr) \leq 1/n\,.\]
\end{lemma}

%
This lemma says that the number of visits to each node can be
bounded. However, for each node, we are only interested in the case
where it is used as a connector. The lemma below shows that the
number of visits as a connector can be bounded as well; i.e.,
if any node $v_i$ appears $t$ times in the walk, then it is likely
to appear roughly $t/\lambda$ times as connectors.

\begin{lemma}
\label{lem:uniformityused} For any vertex $v$, if $v$ appears in the
walk at most $t$ times then it appears as a connector node at most
$t(\log n)^2/\lambda$ times with probability at least $1-1/n^2$.
\end{lemma}

Intuitively, this argument is simple, since the connectors are
spread out in steps of length approximately $\lambda$. However,
there might be some {\em periodicity} that results in the same node
being visited multiple times but {\em exactly} at
$\lambda$-intervals. This is where we crucially use the fact that
the algorithm uses walks of length $\lambda + r$ where $r$ is chosen
uniformly at random from $[0,\lambda-1]$. The proof then goes via constructing another process equivalent to partitioning the $\ell$ steps in to intervals of $\lambda$ and then sampling points from each interval. We analyze this by carefully constructing a different process that stochastically dominates the process of a node occurring as a connector at various steps in the $\ell$-length walk and then use a Chernoff bound argument. The detailed proof is presented in Appendix~\ref{sec:uniformityused-proof}.

Now we are ready to prove Theorem~\ref{thm:1-walk}.

\begin{proof}[Proof of Theorem~\ref{thm:1-walk}]
First, we claim, using Lemma \ref{lemma:visits bound} and
\ref{lem:uniformityused}, that each node is used as a connector node
at most $\frac{24 d(x) \sqrt{\ell}(\log n)^3}{\lambda}$ times with
probability at least $1-2/n$. To see this, observe that the claim
holds if each node $x$ is visited at most
$t(x)=24d(x)\sqrt{\ell+1}\log n$ times and consequently appears as a
connector node at most $t(x)(\log n)^2/\lambda$ times. By
Lemma~\ref{lemma:visits bound}, the first condition holds with
probability at least $1-1/n$. By Lemma~\ref{lem:uniformityused} and
the union bound over all nodes, the second condition holds with
probability at least $1-1/n$, provided that the first condition
holds. Therefore, both conditions hold together with probability at
least $1-2/n$ as claimed.

Now, we choose $\eta=1$ and $\lambda=24 \sqrt{\ell D}(\log n)^3$.
By Lemma~\ref{lem:phase1}, Phase~1 finishes in $\tilde O(\lambda
\eta) = \tilde O(\sqrt{\ell D})$ rounds with high probability.
For Phase~2, {\sc Sample-Destination} is invoked
$O(\frac{\ell}{\lambda})$ times (only when we stitch the walks) and
therefore, by Lemma~\ref{lem:Sample-Destination}, contributes
$O(\frac{\ell D}{\lambda})=\tilde O(\sqrt{\ell D})$ rounds.
Finally, we claim that {\sc Get-More-Walks} is never invoked, with
probability at least $1-2/n$. To see this, recall our claim above
that each node is used as a connector node at most $\frac{24 d(x)
\sqrt{\ell}(\log n)^3}{\lambda}$ times. Moreover, observe that we
have prepared this many walks in Phase~1; i.e., after Phase~1, each
node has $\eta\lambda d(x)= \frac{24 d(x) \sqrt{\ell}(\log
n)^3}{\lambda}$ short walks. The claim follows.

Therefore, with probability at least $1-2/n$, the rounds are $\tilde
O(\sqrt{\ell D})$ as claimed.
\end{proof}

\noindent{\bf Regenerating the entire random walk:} It is important
to note that our algorithm can be extended to regenerate the entire
walk. As described above, the source node obtains the sample after a
random walk of length $\ell$. In certain applications, it may be
desired that the entire random walk be obtained, i.e., every node in
the $\ell$ length walk knows its position(s) in the walk. This can
be done by first informing all intermediate connecting nodes of
their position (since there are only $O(\sqrt{\ell})$ such nodes).
Then, these nodes can regenerate their $O(\sqrt{\ell})$ length short
walks by simply sending a message through each of the corresponding
short walks. This can be completed in $\tilde{O}(\sqrt{\ell D})$
rounds with high probability. This is because, with high
probability, {\sc Get-More-Walk} will not be invoked and hence all
the short walks are generated in Phase~1. Sending a message through
each of these short walks (in fact, sending a message through {\em
every} short walk generated in Phase~1) takes time at most the time
taken in Phase~1, i.e., $\tilde{O}(\sqrt{\ell D})$ rounds.

\subsection{Extension to Computing $k$ Random Walks}

We now consider the scenario when we want to compute $k$ walks of
length $\ell$ from different (not necessary distinct) sources $s_1,
s_2, \ldots, s_k$. We show that {\sc
Single-Random-Walk} can be extended to solve this problem. Consider
the following  algorithm.

\paragraph{{\sc Many-Random-Walks}:} Let
$\lambda=(24 \sqrt{k\ell D+1}\log n+k)(\log n)^2$ and $\eta=1$. If
$\lambda> \ell$ then run the naive random walk algorithm, i.e., the
sources find walks of length $\ell$ simultaneously by sending
tokens. Otherwise, do the following. First, modify Phase~2 of {\sc
Single-Random-Walk} to create multiple walks, one at a time; i.e.,
in the second phase, we stitch the short walks together to get a
walk of length $\ell$ starting at $s_1$ then do the same thing for
$s_2$, $s_3$, and so on. We state the theorem below and the proof is placed in Appendix~\ref{app:kwalks}.

\begin{theorem}\label{thm:kwalks} {\sc Many-Random-Walks} finishes in
$\tilde O\left(\min(\sqrt{k\ell D}+k, k+\ell)\right)$
rounds with high probability.
\end{theorem}



\section{Lower bound}
\label{sec:lowerbound}

In this section, we show an almost tight lower bound on the time
complexity of performing a distributed random walk. At the end of the walk,
we require that each node in the walk should know its correct position(s) among the $\ell$ steps.
We show that any distributed algorithm needs  at least
$\Omega\left(\sqrt{\frac{\ell}{\log \ell}}\right)$ rounds, even in
graphs with low diameter. Note that $\Omega(D)$ is a  lower bound
\cite{DNP09-podc}. Also note that if a source node wants to sample
$k$ destinations from independent random walks, then $\Omega(k)$ is
also a lower bound as the source may need to receive $\Omega(k)$
distinct messages. Therefore, for $k$ walks, the lower bound we show
is $\Omega(\sqrt{\frac{\ell}{\log \ell}} + k + D)$ rounds. 
(The rest of the section omits the $\Omega(k+D)$ term.) In particular, we
show that there exists a $n$-node graph of diameter $O(\log n)$ such
that any distributed algorithm needs at least
$\Omega(\sqrt{\frac{n}{\log n}})$ time to perform a walk of length
$n$. Our lower bound proof makes use of a lower bound for another
problem that we call as the {\em Path Verification problem} defined
as follows. Informally, the Path Verification problem is for
some node $v$ to verify that a
given sequence of nodes in the graph is a valid path of length $\ell$.

\begin{definition}[{\sc Path-Verification} Problem]
The input of the problem consists of an integer $\ell$, a graph $G =
(V,E)$, and $\ell$  nodes $v_1, v_2, ...,
v_\ell$ in $G$. To be precise, each node $v_i$ initially has its
order number $i$.

The goal is for some node $v$ to ``verify" that the above sequence
of vertices forms an $\ell$-length path, i.e., if
$(v_i,v_{i+1})$ forms an edge for all $1 \leq i \leq \ell - 1$.
Specifically, $v$ should output ``yes" if the sequence forms an $\ell$-length
path and ``no" otherwise.
\end{definition}


We show a lower bound for the Path Verification problem that applies
to a very general class of verification algorithms defined as follows.
Each node can (only) verify a segment of the path that it knows
either directly or indirectly (by learning form its neighbors), as
follows.
Initially each node knows only the trivial segment (i.e. the vertex itself). If a
vertex obtains from its neighbor a segment $[i_1,j_1]$ and it has
already verified segment $[i_2,j_2]$ that overlaps with $[i_1, j_1]$
(say, $i_1 < i_2 < j_1 < j_2$) then it can verify a larger interval
($[i_1,j_2]$). Note that a node needs to only send the endpoints of
the interval that it already verifies (hence larger intervals are
better). (See Figure~\ref{fig:path_verify_definition} in the
Appendix for an example.) The goal of the problem is that, in the end, some node verifies the entire segment $[1,\ell]$. We would like to determine a lower bound
for the running time of any distributed algorithm for the above
problem.

A lower bound for the
Path Verification problem, implies a lower bound for the random walk
problem as well.
 The reason is as follows.
Both problems involve
constructing  a path of some specified length $\ell$. Intuitively,
the former is a simpler problem, since we  are not verifying whether
the local steps are chosen randomly, but just whether the path is
valid and is of length $\ell$. On the other hand, any algorithm for
the random walk problem (including our algorithm of Section
\ref{sec:one_walk_DoS}), also solves the Path Verification problem,
since the path it constructs should be a valid path of length
$\ell$.
It is straightforward to make any distributed algorithm that
computes a random walk  to also verify that  indeed the random walk
is a valid walk of appropriate length.  This is essential for
correctness, as otherwise, an adversary can always change  simply
one edge of the graph and ensure that the walk is wrong.

In the next section we first prove a lower bound for the Path
Verification problem. Then we show the same lower bound holds for the
random walk  problem by giving a reduction.

\subsection{Lower Bound for the Path Verification Problem}

The main result of this section is the following theorem.
\begin{theorem}
\label{thm:mainLB} For every $n$, and $\ell \leq n$ there exists a
graph $G_n$ of $\Theta(n)$ vertices and diameter $O(\log n)$, and a
path $P$ of length $\ell$ such that any algorithm  that solves the
{\sc Path-Verification} problem on $G_n$ and $P$ requires more than
$k$ rounds, where $k=\sqrt{\frac{\ell}{\log \ell}}$.
\end{theorem}

The rest of the section is devoted to proving the above Theorem. We
start by defining $G_n$.

\begin{definition}[Graph $G_{n}$] Let $k'$ be an integer such that
$k$ is a power of $2$ and $k'/2\leq 4k < k'$. Let $n'$ be such that
$n'\geq n$ and $k'$ divides $n'$. We construct $G_n$ having
$(n'+2k'-1)=O(n)$ nodes as follows. First, we construct a path
$P=v_1v_2...v_{n'}$. Second, we construct a binary $T$ having $k'$
leaf nodes. Let $u_1, u_2, ..., u_{k'}$ be its leaves from left to
right. Finally, we connect $P$ with $T$ by adding an edge
$u_iv_{jk'+i}$ for every $i$ and $j$. We will denote the root of $T$
by $x$ and its left and right children by $l$ and $r$ respectively.
Clearly, $G_n$ has diameter $O(\log n)$. We then consider a path of
length $\ell=\Theta(n)$. If required $n$ can always be made larger by connecting dummy vertices to the root of $T$. (The resulting graph $G_n$ is as in
Figure~\ref{fig:graph_construction} in the Appendix.) \qed
\end{definition}



To prove the theorem, let $\mathcal A$ be any algorithm for the {\sc
Path-Verification} problem that solves the problem on $G_n$ in at
most $k'$ rounds. We need some definitions and claims to prove the
 theorem.


\paragraph{\bf Definitions of {\em left/right subtrees} and {\em
breakpoints}.}




Consider a tree $T'$ obtained by deleting all edges in $P$. Notice
that nodes $v_{jk'+i}$, for all $j$ and $i \leq k'/2$ are in the
subtree of $T'$ rooted at $l$ and all remaining points are in the
subtree rooted at $r$. For any node $v$, let $sub(v)$ denote the
subtree rooted at node $v$. (Note that $sub(v)$ also include nodes
in the path $P$.) We denote the set of nodes that are leaves of
$sub(l)$ by $L$ (i.e., $L=sub(l)\cap P$) and the set of nodes that
are leaves in $sub(r)$ by $R$.


{

Since we consider an algorithm that takes at most $k$ rounds,
consider the situation when the algorithm is given $k$ rounds for
{\em free} to communicate only along the edges of the path $P$ at
the beginning.
Since $L$ and $R$ consists of every $k'/2$ vertices in $P$ and
$k'/2> 2k$, there are some nodes unreachable from $L$ by walking on
$P$ for $k$ steps. In particular, all nodes of the form
$v_{jk'+k'/2+k+1}$, for all $j$, are not reachable from $L$. We call
such nodes {\em breakpoints} for $sub(l)$. Similarly all nodes of
the form $v_{jk'+k+1}$, for all $j$, are not reachable from $R$ and
we call them the breakpoints for $sub(r)$. (See
Figure~\ref{fig:breaking-points} in the Appendix.)



\paragraph{\bf Definitions of {\em path-distance} and {\em
covering}.}



For any two nodes $u$ and $v$ in $T'$ (obtained from $G_n$ by
deleting edges in $P$), let $c(u, v)$ be a lowest common ancestor of
$u$ and $v$. We define $path\_dist(u, v)$ to be the number of leaves
of subtree of $T$ rooted at $c(u, v)$. Note that the path-distance
is defined between any pair of nodes in $G_n$ but the distance is
counted using the number of leaves in $T$ (which excludes nodes in
$P$).
(See Figure~\ref{fig:path_distance} in Appendix.)

We also introduce the notion of the path-distance {\em covered} by a
message. For any message $m$, the path-distance covered by $m$ is
the maximum path-distance taken over all nodes that have held the
message $m$. That is, if $m$ covers some nodes $v'_1, v'_2, ...,
v'_k$ then the path-distance covered by $m$ is the number of leaves
in the subtrees of $T$ rooted by $v'_1, v'_2, ..., v'_k$. Note that
some leaves may be in more than one subtrees and they will be
counted only once.
Our construction makes the right and left subtrees have a large
number of break points, as in the following lemma. (Proof can be
found in Appendix~\ref{proof:lem:one}.)
%


\begin{lemma}
\label{lem:one} The number of breakpoints for the left subtree and
for the right subtree are at least $\frac{n}{4k}$ each.
\end{lemma}

The reason we define these breakpoints is to show that the entire
information held by the left subtree has many disjoint intervals,
and same for the right subtree. This then tells us that the left
subtree and the right subtree must {\em communicate} a lot to be
able to merge these intervals by connecting/communicating the break
points.
To argue this, we show that the total path distance (over all
messages) is large, as in the following lemma. (Proof is in
Appendix~\ref{proof:lem:two}.)



\begin{lemma}
\label{lem:two} For algorithm $\mathcal A$ to solve {\sc
Path-Verification} problem, the total path-distance covered by all
messages is at least $n$.
\end{lemma}


These messages can however be communicated using the tree edges as
well. We bound the maximum communication that can be achieved across
$sub(l)$ and $sub(r)$ indirectly by bounding the maximum
path-distance that can be covered in each round. In particular, we
show the following lemma. See Figure~\ref{fig:max_path_cover} and proof in the Appendix.

\begin{lemma}
\label{lem:three} In $k$ rounds, all messages together can cover at
most a path-distance of $O(k^2\log k)$.
\end{lemma}



We now describe the proof of the main theorem using these three
claims.

\begin{proof} [Proof of Theorem~\ref{thm:mainLB}]
Use Lemmas~\ref{lem:two} and~\ref{lem:three} we know that if
$\mathcal A$ solves {\sc Path-Verification}, then it needs to cover
a $path\_dist$ of $n$, but in $k$ rounds it can only cover a
$path\_dist$ of $O(k^2\log k)$. But this is $o(n)$ since
$k=\sqrt{\frac{n}{\log n}}$, contradiction.
\end{proof}

\subsection{Reduction to Random Walk Problem} \label{sec:reduction}
We now discuss how the lower bound for the Path Verification problem
implies the lower bound of the random walk problem. The main
difference between {\sc Path-Verification} problem and the random
walk problem is that in the former we can specify which path to
verify while the latter problem generates different path each time.
We show that the ``bad'' instance ($G_n$ and $P$) in the previous
section can be modified so that with high probability, the generated
random walk is ``hard'' to verify. The theorems below are stated for $\ell$ length walk/path instead of $n$ as above. As previously stated, if it is desired that $\ell$ be $o(n)$, it is always possible to add dummy nodes.



\begin{theorem}
For any $n$, there exists a graph $G_n$ of $\Theta(n)$ vertices and
diameter $O(\log n)$, and $\ell=\Theta(n)$ such that, with high
probability, a random walk of length $\ell$  needs
$\Omega(\sqrt{\frac{\ell}{\log \ell}})$  rounds.
\end{theorem}

\begin{proof}
Theorem~\ref{thm:mainLB} can be generalized to the
case where the path $P$ has infinite capacity, as follows.

\begin{theorem}
\label{thm:general-LB} For any $n$ and $\ell=\Theta(n)$, there exists a graph $G_n$ of
$O(n)$ vertices and diameter $O(\log n)$, and a path $P$ of length
$\ell$ such that any algorithm that solves the {\sc
Path-Verification} problem on $G_n$ and $P$ requires more than $\Omega(\sqrt{\frac{\ell}{\log \ell}})$
rounds, even if edges in $P$ have large capacity (i.e., one can send
larger sized messages in one step).
\end{theorem}

\begin{proof}
This is because the proof of Theorem~\ref{thm:mainLB} only uses the
congestion of edges in the tree $T$ (imposed above $P$) to argue
about the number of rounds.
\end{proof}

Now, we modify $G_n$ to $G'_n$ as follows. Recall that the path $P$
in $G_n$ has vertices $v_1, v_2, ..., v_{n'}$. For each $i=1, 2,
..., n'$, we define the weight of an edge $(v_i,v_{i+1})$ to be
$(2n)^{2i}$ (note that weighted graphs are equivalent to unweighted
multigraphs in our model). By having more weight, these edges have
more capacity as well. However, increasing capacity does not affect
the claim as shown above. Observe that, when the walk is at the node
$v_i$, the probability of walk will take the edge $(v_i,v_{i+1})$ is
at least $1-\frac{1}{n^2}$. Therefore, $P$ is the resulting random
walk with probability at least $1-1/n$. When the random walk path is
$P$, it takes at least $\sqrt{\frac{n}{\log n}}$ rounds to verify,
by Theorem~\ref{thm:general-LB}. This completes the proof.
We remark that this construction requires exponential in $n$ number of edges (multiedges). 
For the distributed computing model, this only translates to a larger bandwidth.
The length $\ell$ is still comparable to the number of nodes.
%
%
%
\end{proof}

\section{Applications}

In this section, we present two applications of our algorithm.

\subsection{A Distributed Algorithm for Random Spanning Tree}
\label{sec:rst}
 We now present an algorithm for generating a random spanning
tree (RST) of an unweighted undirected network in $\tilde{O}(\sqrt{m}D)$ rounds with
high probability.  The approach is to simulate Aldous and Broder's \cite{aldous,broder}
RST algorithm  which is as follows. First, pick one arbitrary node as a root. Then, perform a random walk from the root node until
all nodes are visited. For each non-root node, output the edge that
is used for its first visit.  (That is, for each non-root node $v$,
if the first time $v$ is visited is $t$ then we output the edge $(u,v)$
where $u$ is the node visited at time $t-1$.)
The output edges clearly form a spanning tree and this spanning tree
is shown to come from a uniform distribution among all spanning trees of the graph~\cite{aldous,broder}.
The expected time of this algorithm is the expected cover time of
the graph which is shown to be $O(mD)$ (in the worst case, i.e., for any undirected, unweighted graph) by Aleniunas et
al.~\cite{aleliunas}. 

This algorithm can be simulated on the distributed network by our
random walk algorithm as follows. The algorithm can be viewed in phases. Initially, we pick a root node
arbitrarily and set $\ell=n$. In each phase, we run $\log n$ (different) walks of length
$\ell$ starting from the root node (this takes
$\tilde{O}(\sqrt{\ell D})$ rounds using our distributed random walk algorithm).  If none of the $O(\log n)$ different walks cover all nodes (this can be easily checked in $O(D)$ time), we
double the value of $\ell$ and start a new phase, i.e., perform again $\log n$  walks of length $\ell$. The algorithm continues
until one walk of length $\ell$ covers all nodes. We then use
such walk to construct a random spanning tree: As the result of this
walk, each node knows its position(s) in the walk (cf. Section~\ref{sec:analysis}), i.e., it has a list
of steps in the walk that it is visited. Therefore, each non-root
node can pick an edge that is used in its first visit by
communicating to its neighbors. Thus at the end of the algorithm,
each node can know which of its adjacent edges belong to the output tree.  (An additional $O(n)$ rounds may be
used to deliver the resulting tree to a particular node if needed.)

We now analyze the number of rounds in term of $\tau$, the expected
cover time of the input graph. The algorithm takes $O(\log\tau)$ phases
before $2\tau\leq \ell\leq 4\tau$, and since one of $\log n$ random
walks of length $2\tau$ will cover the input graph with high
probability, the algorithm will stop with $\ell\leq 4\tau$ with high
probability. Since each phase takes $\tilde{O}(\sqrt{\ell D})$ rounds, the total number of rounds is $\tilde
{O}(\sqrt{\tau D})$ with high probability.
Since  $\tau=\tilde{O}(mD)$, we have the following theorem.

\begin{theorem}
The  algorithm described above generates a uniform random spanning tree
in $\tilde O(\sqrt{m}D)$ rounds with high probability.
\end{theorem}

\subsection{Decentralized Estimation of Mixing Time}
\label{sec:mixingtime}
We now present an algorithm to estimate the mixing time of a graph from a specified source. Throughout this section, we assume that the graph is connected and non-bipartite (the conditions under which mixing time is well-defined).
The main idea in estimating the mixing time is, given a source node, to run many random
walks of length $\ell$ using the approach described in the previous
section, and use these to estimate the distribution induced by the $\ell$-length
random walk. We then compare the distribution at length $\ell$, with
the stationary distribution to determine if they are {\em close}, and if not, double $\ell$ and retry. For this approach, one issue that we need to address is how to compare two distributions with few samples efficiently (a well-studied problem). We introduce some definitions before formalizing our approach and theorem.





\begin{definition} [Distribution vector]
Let $\pi_x(t)$ define the probability distribution vector reached after $t$ steps when the initial distribution starts with probability $1$ at node $x$. Let $\pi$ denote the stationary distribution vector.
\end{definition}




\begin{definition} [$\tau^x(\epsilon)$ and $\tau^x_{mix}$, mixing time for source $x$]
Define $\tau^x(\epsilon) = \min t : ||\pi_x(t) - \pi||_1 < \epsilon$. Define $\tau^x_{mix} = \tau^x(1/2e)$.
\end{definition}

The goal is to estimate $\tau^x_{mix}$. Notice that the definition of $\tau^x_{mix}$ is consistent due to the following standard monotonicity property of distributions (proof in th appendix).

\begin{lemma}\label{lem:monotonicity}
$||\pi_x(t+1) - \pi||_1 \leq  ||\pi_x(t) - \pi||_1$.
\end{lemma}

To compare two distributions, we use the technique of Batu
et. al.~\cite{BFFKRW} to determine if the distributions are $\epsilon$-near.
Their result (slightly restated) is summarized in the following
theorem.

\begin{theorem}[\cite{BFFKRW}]\label{thm:batu}
For any $\epsilon$, given $\tilde{O}(n^{1/2}poly(\epsilon^{-1}))$ samples of a distribution $X$
over $[n]$, and a specified distribution $Y$, there is a test that outputs PASS with high probability if $|X-Y|_1\leq \frac{\epsilon^3}{4\sqrt{n}\log n}$, and outputs FAIL with high probability if $|X-Y|_1\geq 6\epsilon$.
\end{theorem}


We now give a very brief description of the algorithm of Batu et. al.~\cite{BFFKRW} to illustrate that it can in fact be simulated on the distributed network efficiently. The algorithm partitions the set of nodes in to buckets based on the steady state probabilities. Each of the $\tilde{O}(n^{1/2}poly(\epsilon^{-1}))$ samples from $X$ now falls in one of these buckets. Further, the actual count of number of nodes in these buckets for distribution $Y$ are counted. The exact count for $Y$ for at most $\tilde{O}(n^{1/2}poly(\epsilon^{-1}))$ buckets (corresponding to the samples) is compared with the number of samples from $X$; these are compared to determine if $X$ and $Y$ are close. We refer the reader to their paper~\cite{BFFKRW} for a precise description.

Our algorithm starts with $\ell=1$ and runs $K=\tilde{O}(\sqrt{n})$ walks of length $\ell$ from the specified source $x$. As the test of comparison with the steady state distribution outputs FAIL (for choice of $\epsilon=1/12e$), $\ell$ is doubled. This process is repeated to identify the largest $\ell$ such that the test outputs FAIL with high probability and the smallest $\ell$ such that the test outputs PASS with high probability. These give lower and upper bounds on the required $\tau^x_{mix}$ respectively. Our resulting theorem is presented below and the proof is placed in the appendix.


\begin{theorem}\label{thm:mixmain}
Given a graph with diameter $D$, a node $x$ can find, in $\tilde{O}(n^{1/2} + n^{1/4}\sqrt{D\tau^x(\epsilon)})$ rounds, a time
$\tilde{\tau}^x_{mix}$ such that $\tau^x_{mix}\leq \tilde{\tau}^x_{mix}\leq \tau^x(\epsilon)$, where $\epsilon = \frac{1}{6912e\sqrt{n}\log n}$.
%
%
\end{theorem}

Suppose our estimate of $\tau^x_{mix}$ is close to the mixing time of the graph defined as $\tau_{mix} = \max_{x}{\tau^x_{mix}}$, then this would allow us to estimate several related quantities. Given a mixing time $\tau_{mix}$, we can approximate the spectral gap ($1-\lambda_2$) and the conductance ($\Phi$) due to the
known relations that $\frac{1}{1-\lambda_2}\leq \tau_{mix}\leq \frac{\log n}{1-\lambda_2}$ and $\Theta(1-\lambda_2)\leq \Phi\leq \Theta(\sqrt{1-\lambda_2})$ as shown in~\cite{JS89}.

\section{Concluding Remarks}\label{sec:conclusion}


This paper makes progress towards resolving the time complexity of distributed computation of random walks
in undirected networks. The dependence on the diameter $D$ is still not tight, and it would be interesting to settle this.
There is also a gap in our bounds for performing $k$ independent random walks.
Further, we look at the CONGEST model enforcing a bandwidth restriction and minimize number of rounds. While our algorithms have good {\it amortized} message complexity over several walks,
it would be nice to come up with algorithms that are round efficient and yet have smaller message complexity.

We presented two algorithmic applications of our distributed random walk algorithm: estimating mixing times and computing random spanning trees. It would be interesting to improve upon these results. For example, is there a $\tilde{O}(\sqrt{\tau^x_{mix}} + n^{1/4})$ round algorithm to estimate $\tau^x$; and is there a $\tilde{O}(n)$ round algorithm for RST?

There are several interesting directions to take this work further. Can these techniques be useful for estimating the second eigenvector of the transition matrix (useful for sparse cuts)? Are there efficient distributed algorithms for random walks in directed graphs (useful for PageRank and related quantities)? Finally, from a practical standpoint, it is important to develop algorithms that are robust to failures and it would be nice to extend our techniques to handle such node/edge failures.

  \let\oldthebibliography=\thebibliography
  \let\endoldthebibliography=\endthebibliography
  \renewenvironment{thebibliography}[1]{%
    \begin{oldthebibliography}{#1}%
      \setlength{\parskip}{0ex}%
      \setlength{\itemsep}{0ex}%
  }%
  {%
    \end{oldthebibliography}%
  }
{ \small
\bibliographystyle{abbrv}
\bibliography{Distributed-RW}

\begin{thebibliography}{10}

\bibitem{aldous}
D.~Aldous.
\newblock A random walk construction of uniform random spanning trees and
  uniform labelled trees.
\newblock {\em SIAM Journal on Discrete Mathematics}, 3(4):450--465, 1990.

\bibitem{aleliunas}
R.~Aleliunas, R.~Karp, R.~Lipton, L.~Lovasz, and C.~Rackoff.
\newblock Random walks, universal traversal sequences, and the complexity of
  maze problems.
\newblock In {\em FOCS}, 1979.

\bibitem{AAKKLT}
N.~Alon, C.~Avin, M.~Kouck{\'y}, G.~Kozma, Z.~Lotker, and M.~R. Tuttle.
\newblock Many random walks are faster than one.
\newblock In {\em SPAA}, pages 119--128, 2008.

\bibitem{Baala}
H.~Baala, O.~Flauzac, J.~Gaber, M.~Bui, and T.~El-Ghazawi.
\newblock A self-stabilizing distributed algorithm for spanning tree
  construction in wireless ad hoc networks.
\newblock {\em Journal of Parallel and Distributed Computing}, 63(1):97--104,
  2003.

\bibitem{bar-ilan}
J.~Bar-Ilan and D.~Zernik.
\newblock Random leaders and random spanning trees.
\newblock In {\em 3rd International Workshop on Distributed Algorithms (later
  called DISC)}, 1989.

\bibitem{BFFKRW}
T.~Batu, E.~Fischer, L.~Fortnow, R.~Kumar, R.~Rubenfeld, and P.~White.
\newblock Testing random variables for independence and identity.
\newblock In {\em Proc. of the 42nd IEEE Symposium on Foundations of Computer
  Science (FOCS)}, pages 442--451, 2001.

\bibitem{broder}
A.~Broder.
\newblock Generating random spanning trees.
\newblock In {\em FOCS}, 1989.

\bibitem{BBSB04}
M.~Bui, T.~Bernard, D.~Sohier, and A.~Bui.
\newblock Random walks in distributed computing: A survey.
\newblock In {\em IICS}, pages 1--14, 2004.

\bibitem{frieze}
C.~Cooper, A.~Frieze, and T.~Radzik.
\newblock Multiple random walks in random regular graphs.
\newblock In {\em Preprint}, 2009.

\bibitem{AtishGP08}
A.~{Das Sarma}, S.~Gollapudi, and R.~Panigrahy.
\newblock Estimating pagerank on graph streams.
\newblock In {\em PODS}, pages 69--78, 2008.

\bibitem{DNP09-podc}
A.~{Das Sarma}, D.~Nanongkai, and G.~Pandurangan.
\newblock Fast distributed random walks.
\newblock In {\em PODC}, 2009.

\bibitem{dubhashi}
D.~Dubhashi, F.~Grandioni, and A.~Panconesi.
\newblock Distributed algorithms via lp duality and randomization.
\newblock In {\em Handbook of Approximation Algorithms and Metaheuristics}.
  2007.

\bibitem{elkin-survey}
M.~Elkin.
\newblock An overview of distributed approximation.
\newblock {\em ACM SIGACT News Distributed Computing Column}, 35(4):40--57,
  December 2004.

\bibitem{elkin}
M.~Elkin.
\newblock Unconditional lower bounds on the time-approximation tradeoffs for
  the distributed minimum spanning tree problem.
\newblock In {\em Proceedings of Symposium on Theory of Computing (STOC)}, June
  2004.

\bibitem{peleg-mst}
J.~Garay, S.~Kutten, and D.~Peleg.
\newblock A sublinear time distributed algorithm for minimum-weight spanning
  trees.
\newblock {\em SIAM J. Comput.}, 27:302--316, 1998.

\bibitem{mihail-topaware}
C.~Gkantsidis, G.~Goel, M.~Mihail, and A.~Saberi.
\newblock Towards topology aware networks.
\newblock In {\em IEEE INFOCOM}, 2007.

\bibitem{goyal}
N.~Goyal, L.~Rademacher, and S.~Vempala.
\newblock Expanders via random spanning trees.
\newblock In {\em SODA}, 2009.

\bibitem{JS89}
M.~Jerrum and A.~Sinclair.
\newblock Approximating the permanent.
\newblock {\em SIAM Journal of Computing}, 18(6):1149--1178, 1989.

\bibitem{kelner-madry}
J.~Kelner and A.~Madry.
\newblock Faster generation of random spanning trees.
\newblock In {\em IEEE FOCS}, 2009.

\bibitem{kempe}
D.~Kempe and F.~McSherry.
\newblock A decentralized algorithm for spectral analysis.
\newblock {\em Journal of Computer and System Sciences}, 74(1):70--83, 2008.

\bibitem{khan-podc}
M.~Khan, F.~Kuhn, D.~Malkhi, G.~Pandurangan, and K.~Talwar.
\newblock Efficient distributed approximation algorithms via probabilistic tree
  embeddings.
\newblock In {\em Proc. 27th ACM Symp. on Principles of Distributed Computing
  (PODC)}, 2008.

\bibitem{khan-disc}
M.~Khan and G.~Pandurangan.
\newblock A fast distributed approximation algorithm for minimum spanning
  trees.
\newblock {\em Distributed Computing}, 20:391--402, 2008.

\bibitem{kutten-domset}
S.~Kutten and D.~Peleg.
\newblock Fast distributed construction of k-dominating sets and applications.
\newblock {\em J. Algorithms}, 28:40--66, 1998.

\bibitem{lynch}
N.~Lynch.
\newblock {\em Distributed Algorithms}.
\newblock Morgan Kaufmann Publishers, San Mateo, CA, 1996.

\bibitem{Lyons}
R.~Lyons.
\newblock Asymptotic enumeration of spanning trees.
\newblock {\em Combinatorics, Probability {\&} Computing}, 14(4):491--522,
  2005.

\bibitem{MU-book-05}
M.~Mitzenmacher and E.~Upfal.
\newblock {\em Probability and Computing: Randomized Algorithms and
  Probabilistic Analysis}.
\newblock Cambridge University Press, New York, NY, USA, 2005.

\bibitem{MP}
S.~Muthukrishnan and G.~Pandurangan.
\newblock The bin-covering technique for thresholding random geometric graph
  properties.
\newblock In {\em ACM SODA}, 2005.
\newblock Journal version to appear in Journal of Computer and System Sciences.

\bibitem{PK09}
G.~Pandurangan and M.~Khan.
\newblock Theory of communication networks.
\newblock In {\em Algorithms and Theory of Computation Handbook, Second
  Edition}. CRC Press, 2009.

\bibitem{peleg}
D.~Peleg.
\newblock {\em Distributed computing: a locality-sensitive approach}.
\newblock Society for Industrial and Applied Mathematics, Philadelphia, PA,
  USA, 2000.

\bibitem{peleg-bound}
D.~Peleg and V.~Rabinovich.
\newblock A near-tight lower bound on the time complexity of distributed mst
  construction.
\newblock In {\em Proc. of the 40th IEEE Symp. on Foundations of Computer
  Science}, pages 253--261, 1999.

\bibitem{ST08}
R.~Sami and A.~Twigg.
\newblock Lower bounds for distributed markov chain problems.
\newblock {\em CoRR}, abs/0810.5263, 2008.

\bibitem{Vitter85}
J.~S. Vitter.
\newblock Random sampling with a reservoir.
\newblock {\em ACM Trans. Math. Softw.}, 11(1):37--57, 1985.
\newblock Also appeared in FOCS'83.

\bibitem{ZS06}
M.~Zhong and K.~Shen.
\newblock Random walk based node sampling in self-organizing networks.
\newblock {\em Operating Systems Review}, 40(3):49--55, 2006.

\end{thebibliography}
}

\newpage
\section*{Appendix}
\appendix
\section{Omitted Proofs of Section~\ref{sec:one_walk_DoS} (Upper Bound)}

\subsection{Algorithm descriptions}\label{sec:pseudocode}
The main algorithm for performing a single random walk is described
in {\sc Single-Random-Walk} (cf.
Algorithm~\ref{alg:single-random-walk}). This algorithm, in turn,
uses {\sc Get-More-Walks} (cf. \ref{alg:Get-More-Walks} and {\sc
Sample-Destination} (cf. \ref{alg:Sample-Destination}).

Notice that in Line~\ref{line:reservoir} in
Algorithm~\ref{alg:Get-More-Walks}, the walks of length $\lambda$
are extended further to walks of length $\lambda+r$ where $r$ is a
random number in the range $[0,\lambda-1]$. We do this by extending
the $\lambda$-length walks further, and probabilistically stopping
each walk in each of the next $i$ steps (for $0\leq i\leq
\lambda-1$) with probability $\frac{1}{\lambda-i}$. The reason it
needs to be done this way is because if we first sampled $r$,
independently for each walk, in the range $[0,\lambda-1]$ and then
extended each walk accordingly, the algorithm would need to pass $r$
independently for each walk. This will cause congestion along the
edges; no congestion occurs in the mentioned algorithm as only the
{\em count} of the number of walks along an edge are passed to the
node across the edge.

\newcommand{\mindegree}[0]{\delta}
\begin{algorithm}
\caption{\sc Single-Random-Walk($s$, $\ell$)}
\label{alg:single-random-walk}
\textbf{Input:} Starting node $s$, and desired walk length $\ell$.\\
\textbf{Output:} Destination node of the walk outputs the ID of
$s$.\\

\textbf{Phase 1: (Each node $v$ performs $\eta_v=\eta \deg(v))$
random walks of length $\lambda + r_i$ where $r_i$ (for each $1\leq
i\leq \eta$) is chosen independently at random in the range
$[0,\lambda-1]$.)}
\begin{algorithmic}[1]
\STATE Let $r_{max} = \max_{1\leq i\leq \eta}{r_i}$, the random
numbers chosen independently for each of the $\eta_x$ walks.

\STATE Each node $x$ constructs $\eta_x$ messages containing its ID
and in addition, the $i$-th message contains the desired walk length
of $\lambda + r_i$.

\FOR{$i=1$ to $\lambda + r_{max}$}

\STATE This is the $i$-th iteration. Each node $v$ does the
following: Consider each message $M$ held by $v$ and received in the
$(i-1)$-th iteration (having current counter $i-1$). If the message
$M$'s desired walk length is at most $i$, then $v$ stored the ID of
the source ($v$ is the desired destination). Else, $v$ picks a
neighbor $u$ uniformly at random and forward $M$ to $u$ after
incrementing its counter.

\COMMENT{Note that any iteration could require more than 1 round.}

\ENDFOR

\end{algorithmic}

\textbf{Phase 2: (Stitch $\Theta(\ell/\lambda)$ walks, each of
length in $[\lambda,2\lambda-1]$)}
\begin{algorithmic}[1]
\STATE The source node $s$ creates a message called ``token'' which
contains the ID of $s$

\STATE The algorithm generates a set of {\em connectors}, denoted by
$C$, as follows.

\STATE Initialize $C = \{s\}$

\WHILE {Length of walk completed is at most $\ell-2\lambda$}

  \STATE Let $v$ be the node that is currently holding the token.

  \STATE $v$ calls {\sc Sample-Destination($v$)} and let $v'$ be the
  returned value (which is a destination of an unused random walk starting at $v$
  of length between $\lambda$ and $2\lambda-1$.)

  \IF{$v'$ = {\sc null} (all walks from $v$ have already been used up)}

  \STATE $v$ calls {\sc Get-More-Walks($v$, $\lambda$)} (Perform $\Theta(l/\lambda)$ walks
  of length $\lambda$ starting at $v$)

  \STATE $v$ calls {\sc Sample-Destination($v$)} and let $v'$ be the
  returned value

  \ENDIF

  \STATE $v$ sends the token to $v'$

  \STATE $C = C \cup \{v\}$

\ENDWHILE

\STATE Walk naively until $\ell$ steps are completed (this is at
most another $2\lambda$ steps)

\STATE A node holding the token outputs the ID of $s$

\end{algorithmic}

\end{algorithm}

\begin{algorithm}[t]
\caption{\sc Get-More-Walks($v$, $\lambda$)}
\label{alg:Get-More-Walks} (Starting from node $v$,  perform
$\lfloor\ell/\lambda\rfloor$ number of random walks, each of  length
$\lambda + r_i$ where
$r_i$ is chosen uniformly at random in the range $[0,\lambda-1]$ for the $i$-th walk.) \\
\begin{algorithmic}[1]
\STATE The node $v$ constructs $\lfloor\ell/\lambda\rfloor$
(identical) messages containing its ID.

\FOR{$i=1$ to $\lambda$}

\STATE Each node $u$ does the following:

\STATE - For each message $M$ held by $u$, pick a neighbor $z$
uniformly at random as a receiver of $M$.

\STATE - For each neighbor $z$ of $u$, send ID of $v$ and the number
of messages that $z$ is picked as a receiver, denoted by $c(u, v)$.

\STATE - For each neighbor $z$ of $u$, upon receiving ID of $v$ and
$c(u, v)$, constructs $c(u, v)$ messages, each contains the ID of
$v$.

\ENDFOR

\COMMENT {Each walk has now completed $\lambda$ steps. These walks
are now extended probabilistically further by $r$ steps where each
$r$ is independent and uniform in the range $[0,\lambda-1]$.}

\FOR{$i=0$ to $\lambda-1$}

\STATE \label{line:reservoir} For each message, independently with
probability $\frac{1}{\lambda-i}$, stop sending the message further
and save the ID of the source node (in this event, the node with the
message is the destination). For messages $M$ that are not stopped,
each node picks a neighbor correspondingly and sends the messages
forward as before.

\ENDFOR

\STATE At the end, each destination knows the source ID as well as
the length of the corresponding walk.

\end{algorithmic}

\end{algorithm}

\begin{algorithm}[t]
\caption{\sc Sample-Destination($v$)} \label{alg:Sample-Destination}
\textbf{Input:} Starting node $v$.\\
\textbf{Output:} A node sampled from among the stored
walks (of length in $[\lambda, 2\lambda-1]$) from $v$. \\

\textbf{Sweep 1: (Perform BFS tree)}
\begin{algorithmic}[1]

\STATE Construct a Breadth-First-Search (BFS) tree rooted at $v$.
While constructing, every node stores its parent's ID. Denote such
tree by $T$.

\end{algorithmic}

\textbf{Sweep 2: (Tokens travel up the tree, sample as you go)}
\begin{algorithmic}[1]

\STATE We divide $T$ naturally into levels $0$ through $D$ (where
nodes in level $D$ are leaf nodes and the root node $s$ is in level
$0$).

\STATE Tokens are held by nodes as a result of doing walks of length
between $\lambda$ and $2\lambda-1$ from $v$ (which is done in either
Phase~1 or {\sc Get-More-Walks} (cf.
Algorithm~\ref{alg:Get-More-Walks})) A node could have more than one
token.

\STATE Every node $u$ that holds token(s) picks one token, denoted
by $d_0$, uniformly at random and lets $c_0$ denote the number of
tokens it has.

\FOR{$i=D$ down to $0$}

\STATE Every node $u$ in level $i$ that either receives token(s)
from children or possesses token(s) itself do the following.

\STATE Let $u$ have tokens $d_0, d_1, d_2, \ldots, d_q$, with counts
$c_0, c_1, c_2, \ldots, c_q$ (including its own tokens). The node
$v$ samples one of $d_0$ through $d_q$, with probabilities
proportional to the respective counts. That is, for any $1\leq j\leq
q$, $d_j$ is sampled with probability
$\frac{c_j}{c_0+c_1+\ldots+c_q}$.

\STATE The sampled token is sent to the parent node (unless already
at root), along with a count of $c_0+c_1+\ldots+c_q$ (the count
represents the number of tokens from which this token has been
sampled).

\ENDFOR

\STATE The root output the ID of the owner of the final sampled
token. Denote such node by $u_d$.

\end{algorithmic}

\textbf{Sweep 3: (Go and delete the sampled destination)}
\begin{algorithmic}[1]

\STATE $v$ sends a message to $u_d$ (e.g., via broadcasting). $u_d$
deletes one token of $v$ it is holding (so that this random walk of
length $\lambda$ is not reused/re-stitched).
\end{algorithmic}

\end{algorithm}


\subsection{Proofs of Lemma~\ref{lem:phase1}, \ref{lem:get-more-walks}, \ref{lem:Sample-Destination} and
\ref{lem:correctness-sample-destination-new}}\label{sec:four-proofs}

\begin{proof}[Proof of Lemma~\ref{lem:phase1}]
This proof is a slight modification of the proof of Lemma~2.2 in
\cite{DNP09-podc}, where it is shown that each node can perform
$\eta$ walks of length $\lambda$ together in $O(\lambda \eta
\log{n})$ rounds with high probability. We extend this to the
following statement.
\begin{quote}
Each node $v$ can in fact perform $\eta\deg(v)$ of length $2\lambda$
and still finish in $O(\lambda \eta \log{n})$ rounds.
\end{quote}
The desired claim will follow immediately because each node $v$
performs $\eta\deg(v)$ of length \textit{at most} $\lambda$ in
Phase~1.

Consider the case when each node $v$ creates $\eta \deg(v)\geq \eta$
messages. For each message $M$, any $j=1, 2, ..., \lambda$, and any
edge $e$, we define $X_M^j(e)$ to be a random variable having value
1 if $M$ is sent through $e$ in the $j^{th}$ iteration (i.e., when
the counter on $M$ has value $j-1$). Let $X^j(e)=\sum_{M:
\text{message}} X_M^j(e)$.  We compute the expected number of
messages that go through an edge, see claim below.

\begin{claim}
\label{claim:first} For any edge $e$ and any $j$,
$\mathbb{E}[X^j(e)]=2\eta$.
\end{claim}
\begin{proof}
Assume that each node $v$ starts with $\eta \deg(v)$ messages. Each
message takes a random walk. We prove that after any given number of
steps $j$, the expected number of messages at node $v$ is still
$\eta \deg(v)$.  Consider the random walk's probability transition
matrix, call it $A$. In this case $Au = u$ for the vector $u$ having
value $\frac{\deg(v)}{2m}$ where $m$ is the number of edges in the
graph (since this $u$ is the stationary distribution of an
undirected unweighted graph). Now the number of messages we started
with at any node $i$ is proportional to its stationary distribution,
therefore, in expectation, the number of messages at any node
remains the same.

To calculate $\mathbb{E}[X^j(e)]$, notice that edge $e$ will receive
messages from its two end points, say $x$ and $y$. The number of
messages it receives from node $x$ in expectation is exactly the
number of messages at $x$ divided by $\deg(x)$. The claim follows.
\end{proof}

By Chernoff's bound (e.g., in~\cite[Theorem~4.4.]{MU-book-05}), for
any edge $e$ and any $j$,
$$\mathbb{P}[X^j(e)\geq 4\eta\log{n}]\leq 2^{-4\log{n}}=n^{-4}.$$
It follows that the probability that there exists an edge $e$ and an
integer $1\leq j\leq \lambda$ such that $X^j(e)\geq 4\eta\log{n}$ is
at most $|E(G)| \lambda n^{-4}\leq \frac{1}{n}$ since $|E(G)|\leq
n^2$ and $\lambda\leq \ell\leq n$ (by the way we define $\lambda$).

Now suppose that $X^j(e)\leq 4\eta\log{n}$ for every edge $e$ and
every integer $j\leq \lambda$. This implies that we can extend all
walks of length $i$ to length $i+1$ in $4\eta\log{n}$ rounds.
Therefore, we obtain walks of length $\lambda$ in
$4\lambda\eta\log{n}$
rounds as claimed. 
\end{proof}

\begin{proof}[Proof of Lemma~\ref{lem:get-more-walks}]
The argument is exactly the same as the proof of Lemma~2.4 in
\cite{DNP09-podc}. That is, there is no congestion. We only consider
longer walks (length at most $2\lambda-1$ ) this time. The detail of
the proof is as follows.

Consider any node $v$ during the execution of the algorithm. If it
contains $x$ copies of the source ID, for some $x$, it has to pick
$x$ of its neighbors at random, and pass the source ID to each of
these $x$ neighbors. Although it might pass these messages to less
than $x$ neighbors, it sends only the source ID and a {\em count} to
each neighbor, where the count represents the number of copies of
source ID it wishes to send to such neighbor. Note that there is
only one source ID as one node calls {\sc Get-More-Walks} at a time.
Therefore, there is no congestion and thus the algorithm terminates
in $O(\lambda)$ rounds.
\end{proof}

\begin{proof}[Proof of Lemma~\ref{lem:Sample-Destination}]
This proof is exactly the same as the proof of Lemma~2.5 in
\cite{DNP09-podc}.

Constructing a BFS tree clearly takes only $O(D)$ rounds. In the
second phase where the algorithm wishes to {\em sample} one of many
tokens (having its ID) spread across the graph. The sampling is done
while retracing the BFS tree starting from leaf nodes, eventually
reaching the root. The main observation is that when a node receives
multiple samples from its children, it only sends one of them to its
parent. Therefore, there is no congestion. The total number of
rounds required is therefore the number of levels in the BFS tree,
$O(D)$. The third phase of the algorithm can be done by broadcasting
(using a BFS tree) which needs $O(D)$ rounds.
\end{proof}

\begin{proof}[Proof of Lemma~\ref{lem:correctness-sample-destination-new}]
The claim follows from the correctness of {\sc Sample-Destination}
that the algorithm samples a walk uniformly at random and the fact
that the length of each walk is uniformly sampled from the range
$[\lambda,2\lambda-1]$. The first part is proved in Lemma~2.6 in Das
Sarma et al.~\cite{DNP09-podc} and included below for completeness.
We now prove the second part.

To show that each walk length is uniformly sampled from the range
$[\lambda,2\lambda-1]$, note that each walk can be created in two
ways.
\begin{enumerate}
\item It is created in Phase~1. In this case, since we pick the
length of each walk uniformly from the length
$[\lambda,2\lambda-1]$, the claim clearly holds.
\item It is created by {\sc Get-More-Walk}. In this case, the claim holds by the
technique of {\em reservoir} sampling: Observe that after the
$\lambda^{th}$ step of the walk is completed, we stop extending each
walk at any length between $\lambda$ and $2\lambda-1$ uniformly. To
see this, observe that we stop at length $\lambda$ with probability
$1/\lambda$. If the walk does not stop, it will stop at length
$\lambda+1$ with probability $\frac{1}{\lambda-1}$. This means that
the walk will stop at length $\lambda+1$ with probability
$\frac{\lambda-1}{\lambda}\times \frac{1}{\lambda-1} =
\frac{1}{\lambda-1}$. Similarly, it can be argue that the walk will
stop at length $i$ for any $i\in [\lambda, 2\lambda-1]$ with
probability $\frac{1}{\lambda}$.
\end{enumerate}

We now show the proof of Lemma~2.6 (with slight modification) in Das
Sarma et al. for completeness.

\begin{lemma}[Lemma 2.6 in \cite{DNP09-podc}]\label{lem:correctness-sample-destination}
Algorithm {\sc Sample-Destination}($v$) (cf.
Algorithm~\ref{alg:Sample-Destination}), for any node $v$, samples a
destination of a walk starting at $v$ uniformly at random.
\end{lemma}
\begin{proof}
Assume that before this algorithm starts, there are  $t$ (without
loss of generality, let $t > 0$) ``tokens'' containing ID of $v$
stored in some nodes in the network. The goal is to show that {\sc
Sample-Destination} brings one of these tokens to $v$ with uniform
probability. For any node $u$, let $T_u$ be the subtree rooted at
$u$ and let $S_u$ be the set of tokens in $T_u$. (Therefore, $T_v=T$
and $|S_v|=t$.)

We claim that any node $u$ returns a destination to its parent with
uniform probability (i.e., for any tokens $x\in S_u$, $Pr[ u$
returns $x ]$ is $1/|S_u|$ (if $|S_u|>0$)). We prove this by
induction on the height of the tree. This claim clearly holds for
the base case where $u$ is a leaf node. Now, for any non-leaf node
$u$, assume that the claim is true for any of its children.
To be precise, suppose that $u$ receives tokens and counts from $q$
children. Assume that it receives tokens $d_1, d_2, ..., d_q$ and
counts $c_1, c_2, ..., c_q$ from nodes $u_1, u_2, ..., u_q$,
respectively. (Also recall that $d_0$ is the sample of its own
tokens (if exists) and $c_0$ is the number of its own tokens.) By
induction, $d_j$ is sent from $u_j$ to $u$ with probability
$1/|S_{u_j}|$, for any $1\leq j\leq q$. Moreover, $c_j=|S_{u_j}|$
for any $j$. Therefore, any token $d_j$ will be picked with
probability $\frac{1}{|S_{u_j}|}\times \frac{c_j}{c_0+c_1+...c_q} =
\frac{1}{S_u}$ as claimed.

The lemma follows by applying the claim above to $v$.
\end{proof}
\end{proof}


\subsection{Proof of Lemma~\ref{lem:uniformityused}}
\label{sec:uniformityused-proof}
\begin{proof}
Intuitively, this argument is simple, since the connectors are
spread out in steps of length approximately $\lambda$. However,
there might be some {\em periodicity} that results in the same node
being visited multiple times but {\em exactly} at
$\lambda$-intervals. This is where we crucially use the fact that
the algorithm uses walks of length $\lambda + r$ where $r$ is chosen
uniformly at random from $[0,\lambda-1]$.


We prove the lemma using the following two claims.

\begin{claim}
Consider any sequence $A$ of numbers $a_1, ..., a_\ell'$ of length
$\ell'$. For any integer $\lambda'$, let $B$ be a sequence
$a_{\lambda'+r_1}, a_{2\lambda'+r_1+r_2}, ...,
a_{i\lambda'+r_1+...+r_i}, ...$ where $r_i$, for any $i$, is a
random integer picked uniformly from $[0, \lambda'-1]$. Consider
another subsequence of numbers $C$ of $A$ where an element in $C$ is
picked from from ``every $\lambda'$ numbers'' in $A$; i.e., $C$
consists of $\lfloor\ell'/\lambda'\rfloor$ numbers $c_1, c_2, ...$
where, for any $i$, $c_i$ is chosen uniformly at random from
$a_{(i-1)\lambda'+1}, a_{(i-1)\lambda'+2}, ..., a_{i\lambda'}$.
Then, $Pr[C \text{ contains } a_{i_1}, a_{i_2}, ..., a_{i_k}\}] =
Pr[B = \{a_{i_1}, a_{i_2}, ..., a_{i_k}\}]$ for any set $\{a_{i_1},
a_{i_2}, ..., a_{i_k}\}$.
\end{claim}

\begin{proof}
First consider a subsequence $C$ of $A$. Numbers in $C$ are picked
from ``every $\lambda'$ numbers'' in $A$; i.e., $C$ consists of
$\lfloor\ell'/\lambda'\rfloor$ numbers $c_1, c_2, ...$ where, for
any $i$, $c_i$ is chosen uniformly at random from
$a_{(i-1)\lambda'+1}, a_{(i-1)\lambda'+2}, ..., a_{i\lambda'}$.
Observe that $|C|\geq |B|$. In fact, we can say that ``$C$ contains
$B$''; i.e., for any sequence of $k$ indexes $i_1, i_2, ..., i_k$
such that $\lambda'\leq i_{j+1}-i_j\leq 2\lambda'-1$ for all $j$,
$$Pr[B = \{a_{i_1}, a_{i_2}, ..., a_{i_k}\}] = Pr[C \text{ contains
} \{a_{i_1}, a_{i_2}, ..., a_{i_k}\}].$$
To see this, observe that $B$ will be equal to $\{a_{i_1}, a_{i_2},
..., a_{i_k}\}$ only for a specific value of $r_1, r_2, ..., r_k$.
Since each of $r_1, r_2, ..., r_k$ is chosen uniformly at random
from $[1, \lambda']$, $Pr[B = \{a_{i_1}, a_{i_2}, ..., a_{i_k}\}] =
\lambda'^{-k}$.
Moreover, the $C$ will contain $a_{i_1}, a{i_2}, ..., a_{i_k}\}$ if
and only if, for each $j$, we pick $a_{i_j}$ from the interval that
contains it (i.e., from $a_{(i'-1)\lambda'+1}, a_{(i'-1)\lambda'+2},
..., a_{i'\lambda'}$, for some $i'$). (Note that $a_{i_1}, a_{i_2},
...$ are all in different intervals because $i_{j+1}-i_j\geq
\lambda'$ for all $j$.) Therefore, $Pr[C \text{ contains } a_{i_1},
a_{i_2}, ..., a_{i_k}\}]=\lambda'^{-k}$.
\end{proof}

\begin{claim}
Consider any sequence $A$ of numbers $a_1, ..., a_\ell'$ of length
$\ell'$. Consider subsequence of numbers $C$ of $A$ where an element
in $C$ is picked from from ``every $\lambda'$ numbers'' in $A$;
i.e., $C$ consists of $\lfloor\ell'/\lambda'\rfloor$ numbers $c_1,
c_2, ...$ where, for any $i$, $c_i$ is chosen uniformly at random
from $a_{(i-1)\lambda'+1}, a_{(i-1)\lambda'+2}, ...,
a_{i\lambda'}$.. For any number $x$, let $n_x$ be the number of
appearances of $x$ in $A$; i.e., $n_x=|\{i\ |\ a_i=x\}|$. Then, for
any $R\geq 6n_x/\lambda'$, $x$ appears in $C$ more than $R$ times
with probability at most $2^{-R}$.
\end{claim}
\begin{proof}
For $i=1, 2, ..., \lfloor\ell'/\lambda'\rfloor$, let $X_i$ be a 0/1
random variable that is $1$ if and only if $c_i=x$ and
$X=\sum_{i=1}^{\lfloor\ell'/\lambda'\rfloor} X_i$. That is, $X$ is
the number of appearances of $x$ in $C$. Clearly,
$E[X]=n_x/\lambda'$. Since $X_i$'s are independent, we can apply the
Chernoff bound (e.g., in~\cite[Theorem~4.4.]{MU-book-05}): For any
$R\geq 6E[X]=6n_x/\lambda'$,
$$Pr[X\leq R]\geq 2^{-R}.$$
The claim is thus proved.
\end{proof} 

Now we use the claim to prove the lemma. Choose $\ell'=\ell$ and
$\lambda'=\lambda$ and consider any node $v$ that appears at most
$t$ times. The number of times it appears as a connector node is the
number of times it appears in the subsequence $B$ described in the
claim. By applying the claim with $R=t(\log n)^2$, we have that $v$
appears in $B$ more than $t(\log n)^2$ times with probability at
most $1/n^2$ as desired.
\end{proof}



\subsection{Proof of Lemma~\ref{lemma:visits bound}}\label{sec:proof of visits
bound}

%


We
start with the bound of the first and second moment of the number of
visits at each node by each walk.

\begin{proposition}\label{proposition:first and second moment} For
any node $x$, node $y$ and $t = O(m^2)$,
\begin{equation}
\e[N_t^x(y)] \le 8 d(y) \sqrt{t+1}\,, \ \ \ \mbox{{\rm  and }} \ \ \
\e\Bigl[\bigl(N_t^x(y)\bigr)^2\Bigr] \le \e[N_t^x(y)] + 128 \
d^2(y)\  (t+1)\,.
\end{equation}
\end{proposition}

To prove the above proposition, let $P$ denote the transition
probability matrix of such a random walk and let $\pi$ denote the
stationary distribution of the walk, which in this case is simply
proportional to the degree of the vertex, and let $\pi_\m = \min_x
\pi(x)$.

The basic bound we use is the following estimate from Lyons (see
Lemma~3.4 and Remark~4  in \cite{Lyons}). Let $Q$ denote the
transition probability matrix of a chain with self-loop probablity
$\alpha > 0$, and with $c= \min{\{\pi(x) Q(x,y) : x\neq y \mbox{ and }                      Q(x,y)>0\}}\,.$
Note that for a random walk on an undirected graph, $c=\frac{1}{2m}$. For $k >
0$ a positive integer (denoting time) ,

\begin{equation}
\label{kernel_decay} \bigl|\frac{Q^k(x,y)}{\pi(y)} - 1\bigr| \le
\min\Bigl\{\frac{1}{\alpha c \sqrt{k+1}}, \frac{1}{2\alpha^2 c^2
(k+1)} \Bigr\}\,.
\end{equation}

For $k\leq\beta m^2$ for a sufficiently small constant $\beta$, and small $\alpha$, the above can be simplified to the following bound; see
Remark~3 in \cite{Lyons}.
\begin{equation}
\label{one_sided_decay} Q^k(x,y)  \le \frac{4\pi(y)}{c \sqrt{k+1}} =
\frac{4d(y)}{\sqrt{k+1}}\,.
\end{equation}

Note that given a simple random walk on a graph $G$, and a
corresponding matrix  $P$, one can always switch to the lazy version
$Q=(I+P)/2$, and interpret it as a walk on graph $G'$, obtained by
adding  self-loops  to vertices in $G$ so as to double the degree of
each vertex. In the following, with abuse of notation we assume our
$P$ is such a lazy version of the original one.

\begin{proof}
Let $X_0, X_1, \ldots $ describe the random walk, with $X_i$
denoting the position of the walk at time $i\ge 0$, and let
$\bone_A$ denote the indicator (0-1) random variable, which takes
the value 1 when the event $A$ is true. In the following we also use
the subscript $x$ to denote the fact that the probability or
expectation is with respect to starting the walk at vertex $x$.
First the expectation.
\begin{eqnarray*}
\e[N_t^x(y)] & = & \e_x[  \sum_{i=0}^t \bone_{\{X_i=y\}}] = \sum_{i=0}^t P^i(x,y) \\
& \le &  4 d(y) \sum_{i=0}^t \frac{1}{\sqrt{i+1}} , \ \ \mbox{ (using the above inequality  (\ref{one_sided_decay})) } \\
& \le & 8 d(y) \sqrt{t+1}\,.
\end{eqnarray*}


Abbreviating $N^x_t(y)$   as $N_t(y)$, we now compute the second
moment:
\begin{eqnarray*}
\e[N^2_t(y)] & = & \e_x \Bigl[  \bigl(\sum_{i=0}^t \bone_{\{X_i=y\}} \bigr) \bigl(\sum_{j=0}^t \bone_{\{X_j=y\}} \bigr) \Bigr] \\
& = & \e_x\Bigl[  \sum_{i=0}^t \bone_{\{X_i=y\}}  +  2 \sum_{0\le i < j\le t}^t \bone_{\{X_i = y, \ X_j=y\}} \Bigr] \\
& = & \e[N_t(y)] + 2 \sum_{0\le i < j\le t}^t \Pr(X_i = y, \
X_j=y)\,.
\end{eqnarray*}
To bound the second term on the right hand side above, consider for
$0\le i<j $:
\begin{eqnarray*}
\Pr(X_i = y, \ X_j=y)
& = &  \Pr(X_i = y) \ \Pr(X_j = y | X_i = y) \\
& = & P^i(x,y) \ \ P^{j-i}(y,y)\,, \ \ \ \mbox{ due to the Markovian property }\\
& \le & \frac{4 d(y)}{\sqrt{i+1}} \  \ \frac{4
d(y)}{\sqrt{j-i+1}}\,.
 \ \ \mbox{ (using   (\ref{one_sided_decay})) }
\end{eqnarray*}
Thus,
\begin{eqnarray*}
\sum_{0\le i < j \le t}  \Pr(X_i = y, \ X_j=y) & \le &
\sum_{0\le i \le t} \frac{4 d(y)}{\sqrt{i+1}} \ \sum_{0< j-i \le t-i} \frac{4d(y)}{\sqrt{j-i+1}} \\
& = & 16d^2(y) \sum_{0\le i \le t} \frac{1}{\sqrt{i+1}} \ \sum_{0< k \le t-i} \frac{1}{\sqrt{k+1}}\\
& \le & 32 d^2(y)  \sum_{0\le i \le t} \frac{1}{\sqrt{i+1}} \ \sqrt{t-i+1}\\
& \le & 32 d^2(y) \sqrt{t+1}  \sum_{0\le i \le t} \frac{1}{\sqrt{i+1}} \\
& \le & 64 d^2(y) \ (t+1)\,,
\end{eqnarray*}
which yields the  claimed bound on the second moment in the
proposition. 
\end{proof}

Using the above proposition, we bound the number of visits of each
walk at each node, as follows.

\begin{lemma}\label{lemma:whp one walk one node bound}
For $t=O(m^2)$ and any vertex $y \in G$, the random walk
started at $x$ satisfies:
\begin{equation*}
\Pr\bigl(N^x_t(y) \ge  24  \ d(y) \sqrt{t+1}\log n \bigr) \le
\frac{1}{n^2} \,.
\end{equation*}
\end{lemma}
\begin{proof}
First, it follows from the Proposition that
\begin{equation} \Pr\bigl(N^x_t(y) \ge  2\cdot 12 \
d(y) \sqrt{t+1}\bigr) \le \frac{1}{4} \,.\label{eq:simple bound}
\end{equation}
This is done by using the standard Chebyshev argument that for $B >
0$, $\Pr\bigl(N_t(y) \ge  B \bigr) \le \Pr
\bigl(N^2_t(y) \ge  B^2)  \le \frac{\e\bigl(N_t^2(y)\bigr)}{B^2}$.

%
%
For any $r$, let $L^x_r(y)$ be the time that the random walk
(started at $x$) visits $y$ for the $r^{th}$ time. Observe that, for
any $r$, $N^x_t(y)\geq r$ if and only if $L^x_r(y)\leq t$.
Therefore,
\begin{equation}
\Pr(N^x_t(y)\geq r)=\Pr(L^x_r(y)\leq t).\label{eq:visits eq length}
\end{equation}

Let $r^*=24  \ d(y) \sqrt{t+1}$. By \eqref{eq:simple bound} and
\eqref{eq:visits eq length}, $\Pr(L^x_{r^*}(y)\leq t)\leq
\frac{1}{4}\,.$ We claim that
\begin{equation}
\Pr(L^x_{r^*\log n}(y)\leq t)\leq \left(\frac{1}{4}\right)^{\log
n}=\frac{1}{n^2}\,.\label{eq:hp length bound}
\end{equation}
To see this, divide the walk into $\log n$ independent subwalks,
each visiting $y$ exactly $r^*$ times. Since the event $L^x_{r^*\log
n}(y)\leq t$ implies that all subwalks have length at most $t$,
\eqref{eq:hp length bound} follows.
Now, by applying \eqref{eq:visits eq length} again,
\[\Pr(N^x_t(y)\geq r^*\log n) = \Pr(L^x_{r^*\log n}(y)\leq t)\leq
\frac{1}{n^2}\] as desired.

\end{proof}

We now extend the above lemma to bound the number of visits of {\em
all} the walks at each particular node.

\begin{lemma}\label{lemma:k walks one node bound}
For $\gamma > 0$, and $t=O(m^2)$, and for any vertex $y \in
G$, the random walk started at $x$ satisfies:
\begin{equation*}
\Pr\bigl(\sum_{i=1}^k N^{x_i}_t(y) \ge  24  \ d(y) \sqrt{kt+1} \log
n+k\bigr) \le \frac{1}{n^2} \,.
\end{equation*}
\end{lemma}
\begin{proof}
First, observe that, for any $r$, $$\Pr\bigl(\sum_{i=1}^k
N^{x_i}_t(y) \geq r-k\bigr)\leq \Pr[N^y_{kt}(y)\geq r].$$ To see
this, we construct a walk $W$ of length $kt$ starting at $y$ in the
following way: For each $i$, denote a walk of length $t$ starting at
$x_i$ by $W_i$. Let $\tau_i$ and $\tau'_i$ be the first and last
time (not later than time $t$) that $W_i$ visits $y$. Let $W'_i$ be
the subwalk of $W_i$ from time $\tau_i$ to $\tau_i'$. We construct a
walk $W$ by stitching $W'_1, W'_2, ..., W'_k$ together and complete
the rest of the walk (to reach the length $kt$) by a normal random
walk. It then follows that the number of visits to $y$ by $W_1, W_2,
\ldots, W_k$ (excluding the starting step) is at most the number of
visits to $y$ by $W$. The first quantity is $\sum_{i=1}^k
N^{x_i}_t(y)-k$. (The term `$-k$' comes from the fact that we do not
count the first visit to $y$ by each $W_i$ which is the starting
step of each $W'_i$.) The second quantity is $N^y_{kt}(y)$. The
observation thus follows.

Therefore, \[\Pr\bigl(\sum_{i=1}^k N^{x_i}_t(y)\geq 24 \ d(y)
\sqrt{kt+1}\log n + k\bigr) \leq \Pr\bigl(N^y_{kt}(y)\geq 24 \ d(y)
\sqrt{kt+1}\log n\bigr) \leq \frac{1}{n^2}\]
where the last inequality follows from Lemma~\ref{lemma:whp one walk
one node bound}.
%
%
\end{proof}

Lemma~\ref{lemma:visits bound} follows immediately from
Lemma~\ref{lemma:k walks one node bound} by union bounding over all
nodes.

\subsection{Proof of Theorem~\ref{thm:kwalks}}\label{app:kwalks}
\begin{proof}
First, consider the case where $\lambda>\ell$. In this case,
$\min(\sqrt{k\ell D}+k, \sqrt{k\ell}+k+\ell)=\tilde
O(\sqrt{k\ell}+k+\ell)$. By Lemma~\ref{lemma:visits bound}, each
node $x$ will be visited at most $\tilde O(d(x) (\sqrt{k\ell}+k))$
times. Therefore, using the same argument as Lemma~\ref{lem:phase1},
the congestion is $\tilde O(\sqrt{k\ell} + k)$ with high
probability. Since the dilation is $\ell$, {\sc Many-Random-Walks}
takes $\tilde O(\sqrt{k\ell} + k +\ell)$ rounds as claimed. Since $2\sqrt{k\ell}\leq k+\ell$, this bound reduces to $O(k+\ell)$.

Now, consider the other case where $\lambda\leq \ell$. In this case,
$\min(\sqrt{k\ell D}+k, \sqrt{k\ell}+k+\ell)=\tilde O(\sqrt{k\ell
D}+k)$. Phase~1 takes $\tilde O(\lambda \eta) = \tilde O(\sqrt{k\ell
D}+k)$. The stitching in Phase~2 takes $\tilde O(k\ell D/\lambda) =
\tilde O(\sqrt{k\ell D})$. Moreover, by Lemma~\ref{lemma:visits
bound}, {\sc Get-More-Walks} will never be invoked. Therefore, the
total number of rounds is $\tilde O(\sqrt{k\ell D}+k)$ as claimed.
\end{proof}


\section{Omitted Proofs of Section~\ref{sec:lowerbound} (Lower Bound)}

\subsection{Proof of Lemma~\ref{lem:one}}
\label{proof:lem:one}
\begin{proof}
After the first $k$ free rounds, consider the intervals that the
left subtree can have, in the best case. Recall that these $k$
rounds allowed communication only along the path. The $path\_dist$
of any node in $L$ from the breakpoints of $sub(L)$ along the path
is at least $k+1$.
\end{proof}


\subsection{Proof of Lemma~\ref{lem:two}}
\label{proof:lem:two}
\begin{proof}
First, notice that each left breakpoint is at a path-distance of
$k+1$ from every node in the right subtree. That is,
$path\_dist(u,L) = path\_dist(v,R) = k+1$ for all $u\in B_l$ and all
$v\in B_r$.



Each breakpoint needs to be combined into one interval in the end.
However, there could be one interval that is communicated from the
$sub(l)$ to the $sub(r)$ (or vice versa) such that it connects
several breakpoints. We show that this cannot happen. Consider all
the breakpoints $v\in B_l\cup B_r$.

\noindent{\bf Definition of {\em scratching}}.

Let us say that we {\em scratch out} the breakpoints from the list
$k+1$, $k'/2+k+1$, $k'+k+1$, $k'+k'/2+k+1$, $2k'+k+1$, ... that get
connected when an interval is communicated between $sub(l)$ and
$sub(r)$. We scratch out a breakpoint if there is an interval in the
graph that contains it and both (or one in case of the first and
last breakpoints) its adjacent breakpoints. For example, if the left
subtree has intervals $[1, k'/2+k]$ and $[k'/2+k+2, k'+k'/2+k+1]$
and the right subtree has $[k+2, k'+k]$ and the latter interval is
communicated to a node in the left subtree, then the left subtree is
able to obtain the merged interval $[1,k'+k'/2+k+1]$ and therefore
breakpoints $k+1$ and $k'/2+k+1$ are scratched out.

\begin{claim}
\label{claim:one} At most $O(1)$ breakpoints can be scratched out
with one message/interval communicated between $sub(r)$ and $sub(l)$
\end{claim}
\begin{proof}
We argue that with the communication of one interval across the left
and right subtrees, at most $4$ breakpoints that have not been
scratched yet can get scratched. This follows from a simple
inductive argument. Consider a situation where the left subtree has
certain intervals with all overlapping intervals already merged, and
similarly right subtree. Suppose an interval ${\cal I}$ is
communicated between $sub(r)$ and $sub(l)$, one of the following
cases arise:
\squishlist
\item ${\cal I}$ contains one breakpoint: Can be merged with at most two other intervals. Therefore, at most three breakpoints can get scratched.
\item ${\cal I}$ contains two breakpoints: Can get connected with at most two other intervals and therefore at most four breakpoints can get scratched.
\item ${\cal I}$ contains more than two breakpoints: This is impossible since there are at most two breakpoints in each interval, its left most and
right most numbers (by definition of scratching).
\squishend This completes the proof of the claim.
\end{proof}



The proof now follows from Lemma~\ref{lem:one}. For any breakpoint
$b$, let $M_b$ be the set of messages that represents an interval
containing $b$ while $b$ is still unscratched. If $b$ is in $sub(l)$
and gets scratched because of the combination of some intervals in
$sub(r)$, then we claim that $M_b$ has covered a path-distance of at
least $k$. (Define the path-distance covered by $M_b$ by the total
path-distance covered by all messages in $M_b$.) This is because $b
= v_i$ (say), being a breakpoint in $sub(l)$ has $i$ equal to $(k+1
\mod k')$. Therefore, $b$ is at a path distance of at least $k$ from
any node in $R$. Consequently, $b$ is at a path-distance of at least
$k$ from any node in $sub(r)$. Since there are
$\Theta(\frac{n}{4k})$ breakpoints, and for any interval to be
communicated across the left and right subtree, a path-distance of
$k$ must be covered, in total, $\Theta(n)$ path-distance must be
covered for all breakpoints to be scratched. This follows from three
main observations:
\squishlist
\item As shown above, for any breakpoint to be scratched, an interval with a breakpoint must
be communicated from $sub(l)$ to $sub(r)$ or vice versa (thereby all
messages $m$ containing the breakpoint together covering a
path-distance of at least $k$)
\item Any message/interval with unscratched breakpoints has at most two unscratched breakpoints
\item As shown in Claim~\ref{claim:one}, at most four breakpoints can be scratched when two intervals are merged.
\squishend

The proof follows. (Also see Figure~\ref{fig:scratch_4} for the idea of this proof.)
\end{proof}


\subsection{Proof of Lemma~\ref{lem:three}}
\label{proof:lem:three}
\begin{proof}
We consider the total number of messages that can go through nodes
at any level of the graph, starting from level $0$ to level $\log k$
under the congest model.

First notice that if a message is passed at level $i$ of the tree,
this can cover a $path\_dist$ of at most $2^i$. This is because the
subtree rooted at a node at level $i$ has $2^i$ leaves. Further, by
our construction, there are $2^{\log (k') - i}$ nodes at level $i$.
Therefore, all nodes at level $i$ together, in a given round of
$\mathcal A$ can cover a $dist-path$, path distance, of at most
$2^i2^{\log (k') - i} = 4k+2$. Therefore, over $k$ rounds, the total
$path\_dist$ that can be covered in a single level is $k(k')$. Since
there are $O(\log k)$ levels, the total $path\_dist$ that can be
covered in $k$ rounds over the entire graph is $O(k^2\log k)$. (See
Figure~\ref{fig:max_path_cover}.)
\end{proof}


%
%
%

\section{Omitted Proofs of Section~\ref{sec:mixingtime} (Mixing Time)}

\subsection{Brief description of algorithm for Theorem~\ref{thm:batu}}\label{app:batu}

The algorithm partitions the set of nodes in to buckets based on the steady state probabilities. Each of the $\tilde{O}(n^{1/2}poly(\epsilon^{-1}))$ samples from $X$ now falls in one of these buckets. Further, the actual count of number of nodes in these buckets for distribution $Y$ are counted. The exact count for $Y$ for at most $\tilde{O}(n^{1/2}poly(\epsilon^{-1}))$ buckets (corresponding to the samples) is compared with the number of samples from $X$; these are compared to determine if $X$ and $Y$ are close. We refer the reader to their paper~\cite{BFFKRW} for a precise description.

\subsection{Proof of Lemma~\ref{lem:monotonicity}}\label{app:mon}
\begin{proof}
The monotonicity follows from the fact that
$||Ax||_1 \le ||x||_1$ where $A$ is the transpose of the transition probability matrix of the graph and $x$ is any probability vector. That is, $A(i,j)$ denotes the probability of transitioning from node $j$ to node $i$. This in turn follows from the fact that the sum of entries of any column of $A$ is 1.

Now let $\pi$ be the stationary distribution of the transition matrix $A$. This implies that if $\ell$ is $\epsilon$-near mixing, then $||A^lu - \pi||_1 \leq \epsilon$, by definition of $\epsilon$-near mixing time. Now consider $||A^{l+1}u - \pi||_1$. This is equal to $||A^{l+1}u - A\pi||_1$ since $A\pi = \pi$.  However, this reduces to $||A(A^{l}u - \pi)||_1 \leq \epsilon$. It follows that $(\ell+1)$ is $\epsilon$-near mixing.
\end{proof}

\subsection{Proof of Theorem~\ref{thm:mixmain}}\label{app:mixproof}
\begin{proof}
For undirected unweighted graphs, the
stationary distribution of the random walk is known and is
$\frac{deg(i)}{2m}$ for node $i$ with degree $deg(i)$, where $m$ is
the number of edges in the graph.  If a source node in the network knows the degree distribution, we only need
$\tilde{O}(n^{1/2}poly(\epsilon^{-1}))$ samples from a distribution to
compare it to the stationary distribution.  This can be achieved by
running {\sc MultipleRandomWalk} to obtain $K = \tilde{O}(n^{1/2}poly(\epsilon^{-1}))$ random walks. We choose $\epsilon = 1/12e$.
To find the approximate mixing time, we try out
increasing values of $l$ that are powers of $2$.  Once we find the
right consecutive powers of $2$, the monotonicity property admits a
binary search to determine the exact value for the specified $\epsilon$.

The result
in~\cite{BFFKRW} can also be adapted to compare with the steady state distribution even if the source does not know the entire distribution. As described previously, the source only needs to know the {\em count} of number of nodes with steady state distribution in given buckets. Specifically, the buckets of interest are at most $\tilde{O}(n^{1/2}poly(\epsilon^{-1}))$ as the count is required only for buckets were a sample is drawn from. Since each node knows its own steady state probability (determined just by its degree), the source can broadcast a specific bucket information and recover, in $O(D)$ steps, the count of number of nodes that fall into this bucket. Using the standard upcast technique previously described, the source can obtain the bucket count for each of these at most $\tilde{O}(n^{1/2}poly(\epsilon^{-1}))$ buckets in $\tilde{O}(n^{1/2}poly(\epsilon^{-1}) + D)$ rounds.

We have shown previously that a source node can obtain $K$ samples from $K$ independent random walks of length $\ell$ in $\tilde{O}(K + \sqrt{KlD})$ rounds. Setting $K=\tilde{O}(n^{1/2}poly(\epsilon^{-1}) + D)$ completes the proof.
\end{proof}

\section{Figures}

\begin{figure}[h]
\centering
\includegraphics{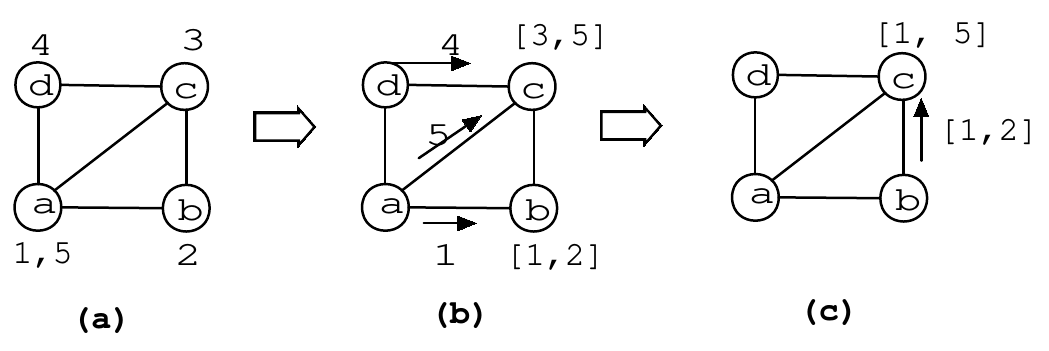}
\caption{Example of path verification problem. {\bf (a)} In the
beginning, we want to verify that the vertices containing numbers
$1..5$ form a path. (In this case, they form a path $a, b, c, d,
a$.) {\bf (b)} One way to do this is for $a$ to send $1$ to $b$ and
therefore $b$ can check that two vertices $a$ and $b$ corresponds to
label $1$ and $2$ form a path. (The interval $[1,2]$ is used to
represent the fact that vertices corresponding to numbers $1, 2$ are
verified to form a path.) Similarly, $c$ can verify $[3,5]$. {\bf
(c)} Finally, $c$ combine $[1,2]$ with $[3, 5]$ and thus the path
corresponds to numbers $1,2, ..., 5$ is verified. }
\label{fig:path_verify_definition}
\end{figure}

\begin{figure}[h]
\centering
\includegraphics[width=0.98\linewidth]{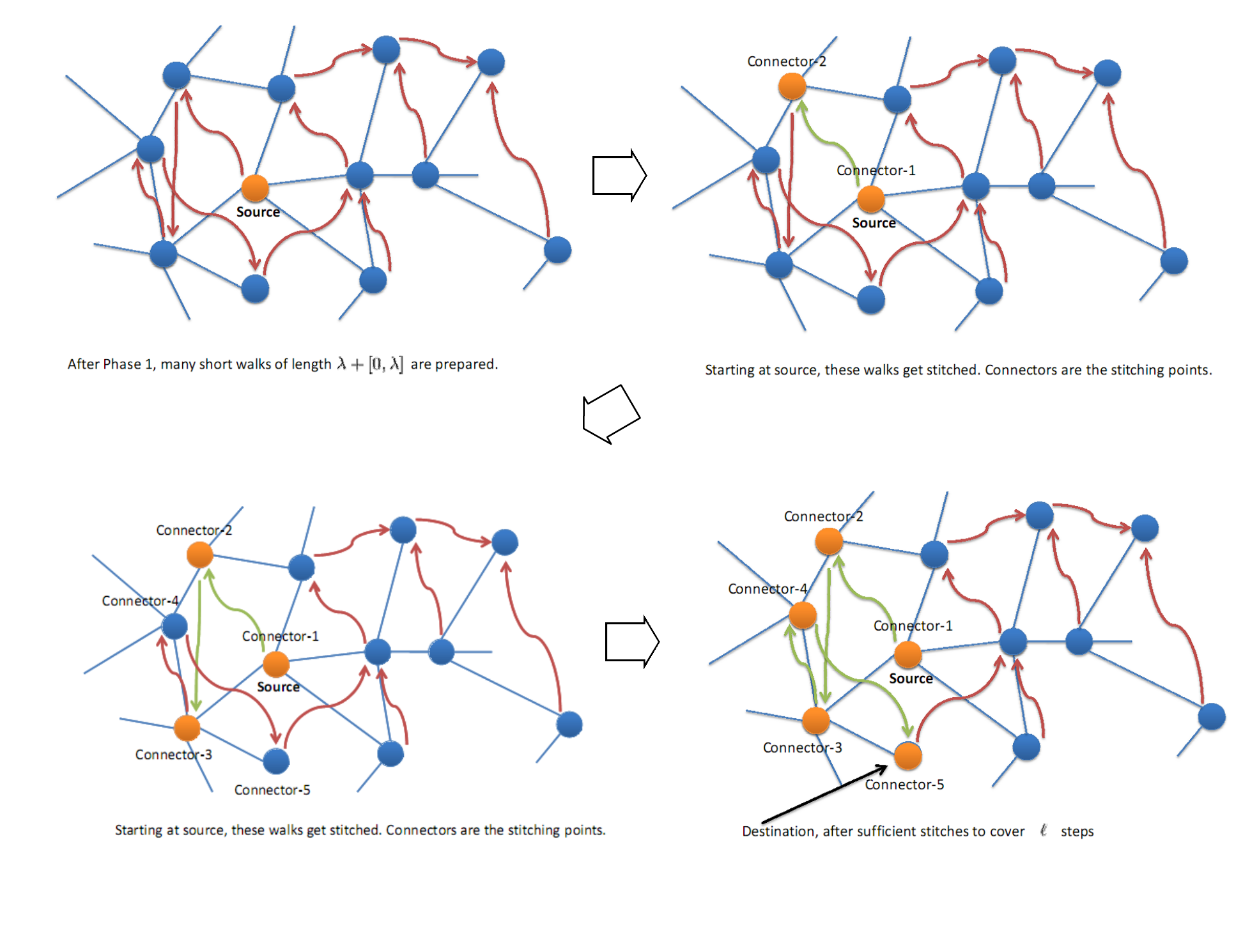}
\caption{Figure illustrating the Algorithm of stitching short walks together.}
\label{fig:connector}
\end{figure}

\begin{figure}[h]
  \centering
  \includegraphics[width=0.7\linewidth]{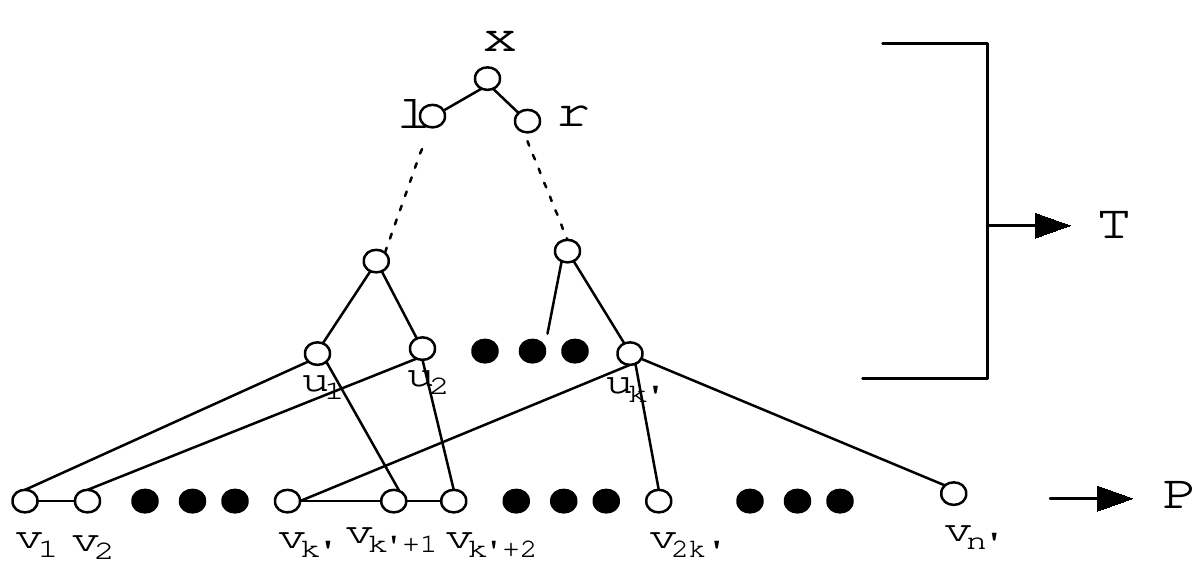}\\
  \caption{$G_n$}\label{fig:graph_construction}
\end{figure}

\begin{figure}[h]
\centering
\includegraphics{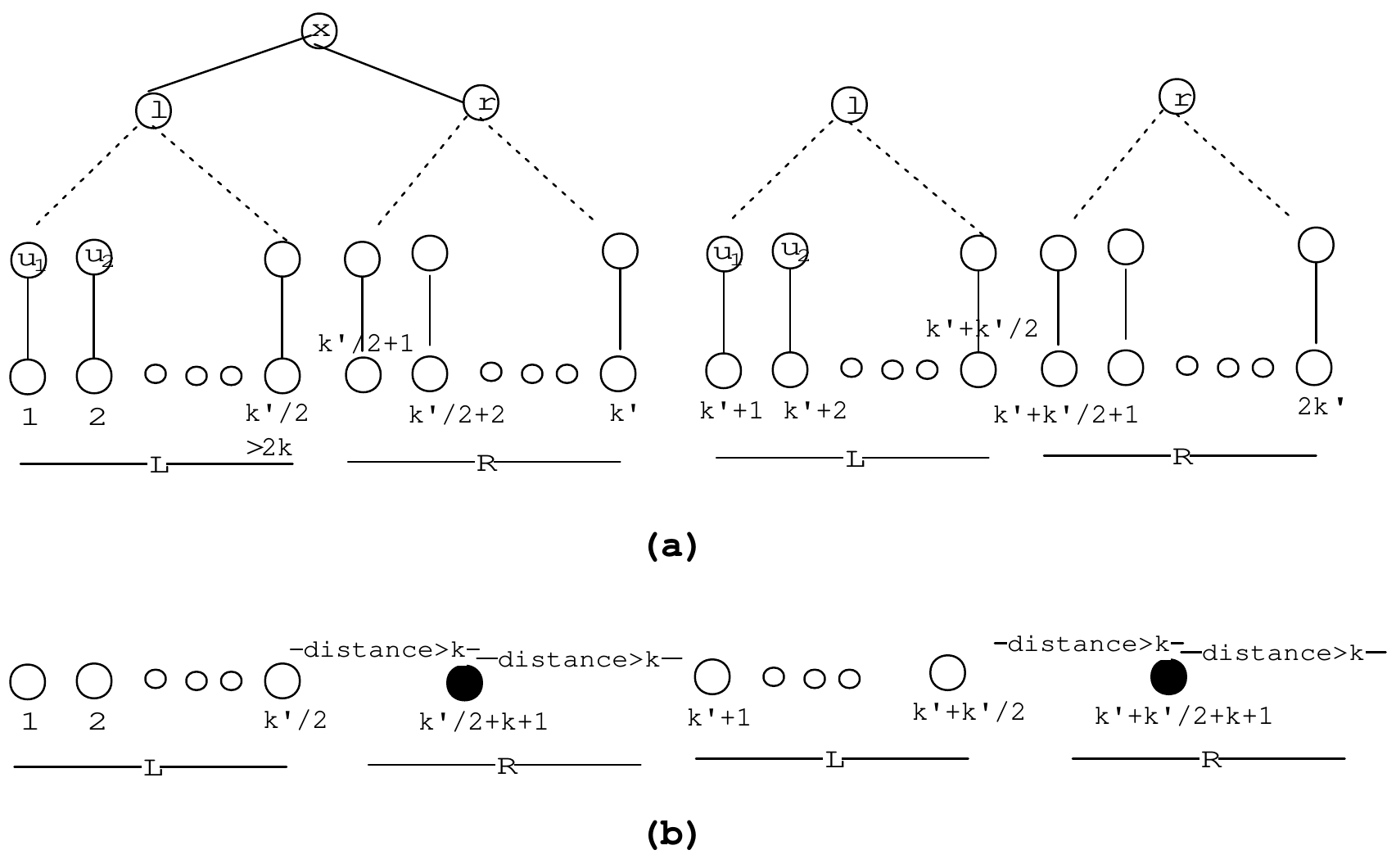}
\caption{\textbf{Breakpoints.} {\bf (a)} $L$ and $R$ consist of
every other $k'/2$ vertices in $P$. (Note that we show the vertices
$l$ and $r$ appear many times for the convenience of presentation.)
{\bf (b)} $v_{k'/2+k+1}$ and $v_{k'+k'/2+k+1}$ (nodes in black) are
two of the breakpoints for $L$. Notice that there is one breakpoint
in every connected piece of $L$ and $R$.}
\label{fig:breaking-points}
\end{figure}

\begin{figure}[h]
\centering \subfigure[{Path-distance.}]{
\includegraphics[width=0.35\linewidth]{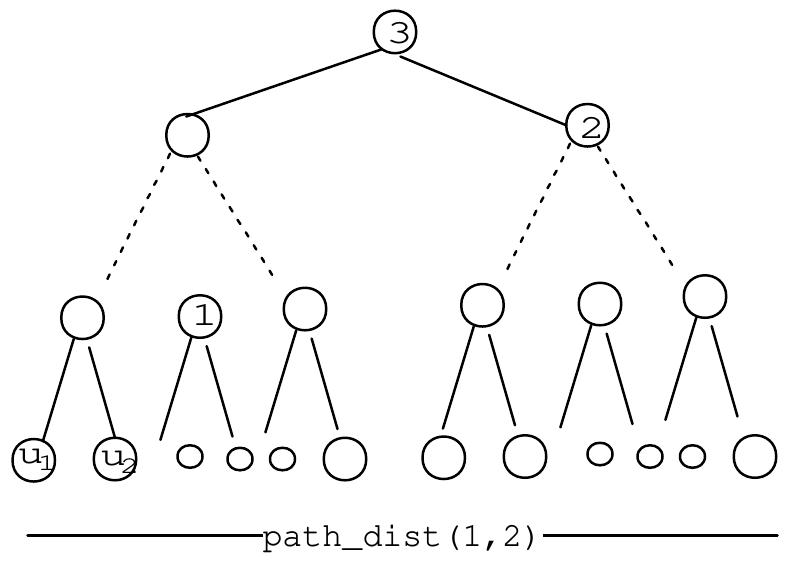}
\label{fig:path_distance}}
\subfigure[{Idea of Claim~\ref{claim:one}}]{
\includegraphics[width=0.35\linewidth]{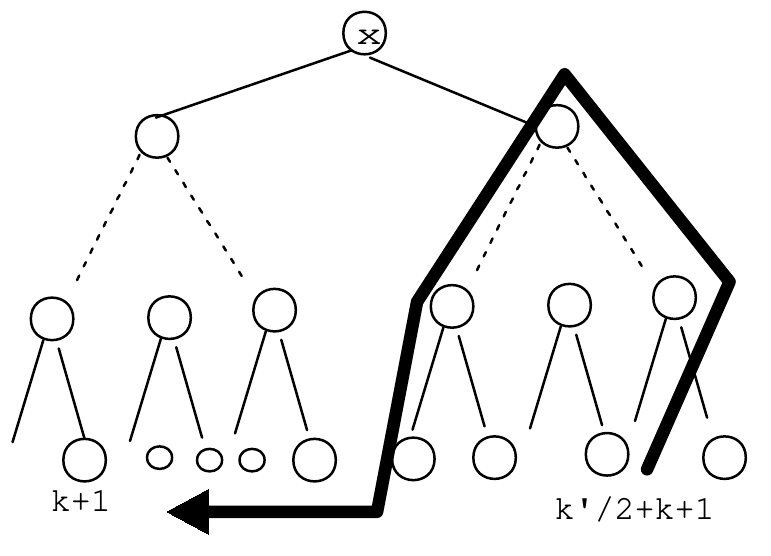}
\label{fig:scratch_4} }
\subfigure[Idea of Lemma~\ref{lem:three}.]{
\includegraphics[width=0.25\linewidth]{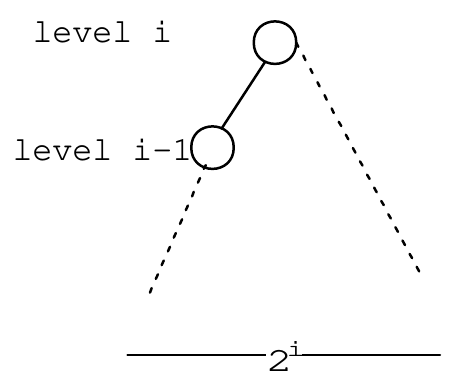}
\label{fig:max_path_cover} } \caption{{\bf (a)} Path distance
between 1 and 2 is the number of leaves in the subtree rooted at 3,
the lowest common ancestor of 1 and 2. {\bf (b)} For one unscratched
left breakpoint, $k'/2+k+1$ to be combined with another right
breakpoint $k+1$ on the left, $k'/2+k+1$ has to be carried to $L$ by
some intervals. Moreover, one interval can carry at most two
unscratched breakpoints at a time. {\bf (c)} Sending a message
between nodes on level $i$ and $i-1$ can increase the covered path
distance by at most $2^i$.}
\end{figure}

\end{document}